\documentclass[11pt, a4paper]{article}%
\usepackage{authblk}
\usepackage{graphicx}
\usepackage{url}
\usepackage{amssymb,mathtools,comment,amsthm,amsmath, thm-restate}
\usepackage{appendix}
\usepackage{color}
\usepackage{multirow}
\usepackage{cite}

\usepackage[backref,colorlinks=true]{hyperref}
\usepackage{bookmark} 

\usepackage{array}

\allowdisplaybreaks

\newtheorem{theorem}{Theorem}
\newtheorem{lemma}{Lemma}
\newtheorem{corollary}{Corollary}

\newtheorem{definition}{Definition}

\theoremstyle{remark}

\usepackage{xspace}
\newcommand{\bigo}[1]{\mathcal{O}#1}
\newcommand{\undecided}{\textsc{Undecided-State}\xspace}
\newcommand{\poly}{poly}

\newcommand{\nat}{\mathbb{N}}

\newcommand{\realnum}{\mathbb{R}}

\newcommand{\prob}{\mathbb{P}}
\newcommand{\pr}[1]{\mathbb{P}\left(#1\right)}
\newcommand{\mean}{\mathbb{E}}
\newcommand{\config}{\mathbf{x}}
\newcommand{\expect}[1]{\mathbb{E}\left[#1 \right]}

\newcommand{\udyn}{\undecided dynamics\xspace}
\newcommand{\uproc}{\undecided process\xspace}
\newcommand{\pruned}{\textsc{Pruned}\xspace process\xspace}

\newcommand{\voter}{\textsc{Voter}\xspace}
\newcommand{\averaging}{\textsc{Averaging}\xspace}

\newcommand{\abs}[1]{\lvert #1 \rvert}

\newcommand{\mesalpha}{\emph{Alpha}\xspace} 
\newcommand{\mesbeta}{\emph{Beta}\xspace}

\newcommand{\stubborn}{\textsc{stub}\xspace}
\newcommand{\stubproc}{\stubborn process\xspace}

\newcommand{\PULL}{Uniform-PULL\xspace}
\newcommand{\PUSH}{Uniform-PUSH\xspace}

\newcommand{\pnoise}{p_{noise}}
\newcommand{\additionalNodes}{n_{stub}}

\begin{document}

\title{Phase Transition of a Non-Linear Opinion Dynamics with Noisy Interactions}

\author[1]{Francesco d'Amore} 
\author[2]{Andrea Clementi} 
\author[1]{Emanuele Natale}

\affil[1]{Universit\'e C\^ote d'Azur, Inria, CNRS, I3S, France. {\tt $\{$francesco.d-amore,emanuele.natale$\}$@inria.fr}}
\affil[2]{University of Rome Tor Vergata, Rome, Italy {\tt clementi@mat.uniroma2.it}}

\date{}

\maketitle              

\begin{abstract}
 
 In   several real \emph{Multi-Agent Systems} (MAS), it has been  
observed that only weaker forms of
\emph{metastable consensus} are achieved, in which  a large majority of agents agree on some  opinion  while other opinions continue to be supported  by  a (small) minority of agents. 
In this work, we take a step towards the investigation of metastable consensus for   complex (non-linear) \emph{opinion dynamics} by considering the famous \undecided dynamics in the binary setting, which is known to reach consensus exponentially faster than the \voter dynamics.   
We propose a simple form of uniform noise in which each message can change to another one with probability $p$ and we prove that the persistence of a \emph{metastable consensus} undergoes a \emph{phase transition} for $p=\frac 16$. In detail, below this threshold,  we prove the system  reaches with high probability a metastable regime where  a large majority of agents keeps supporting the same    opinion  for polynomial time. Moreover, this opinion turns out to be    the   initial majority opinion, whenever the initial bias is slightly  larger than its standard deviation.
On the contrary, above the threshold, we show that the information about the initial majority opinion is  ``lost'' within  logarithmic time  even when the initial bias is maximum.
Interestingly, using a simple coupling argument, we show the equivalence between our noisy model above and the model where a subset of   agents behave in a \emph{stubborn} way.

\end{abstract}

\newpage

\tableofcontents

\section{Introduction}
\medskip

\medskip
We consider a fully-decentralized \emph{Multi-Agent Systems} (for short, MAS) formed by a set of $n$ agents (i.e. nodes) which  mutually interact  by exchanging messages over an underlying communication  graph.
In this setting, 
\emph{opinion dynamics} are mathematical models to  investigate 
the way a fully-decentralized  MAS is able to  reach some form of \emph{Consensus}. 
Their study is a hot topic touching  several research areas such as MAS \cite{CHK18,DNAW00}, Distributed Computing \cite{BCN20,cruciani_distributed_2019,ghaffari_nearly-tight_2018}, Social Networks \cite{acemoglu_opinion_2012,mossel_opinion_2017}, and System Biology \cite{boczkowski_limits_2018,cardelli_cell_2012}.
Typical examples of opinion dynamics are the Voter Model, the averaging rules, and the majority rules.  
Some of such  dynamics   share a surprising efficiency and resiliency that seem to exploit common \emph{computational principles}  
 \cite{BCN20,cruciani_distributed_2019,ghaffari_nearly-tight_2018}. 

Within such framework, the tasks of \emph{(valid) Consensus} and \emph{Majority Consensus} 
 have attracted a lot of attention within different application domains  in social networks \cite{mossel_opinion_2017}, 
 in biological systems \cite{feinerman_breathe_2017}, 
passively-mobile sensor networks \cite{angluin_simple_2007} and chemical reaction networks \cite{condon_approximate_2019}. In the  Consensus task, the system is required to converge to a  stable
configuration where all agents supports the same opinion and this opinion must be \emph{valid}, i.e., it must be supported by at least one agent in the initial configuration. While, in the Majority Consensus task, starting from an initial configuration where there is some positive bias towards one \emph{majority  opinion}, the system is required to converge to the  configuration where all agents support the initial majority opinion. Here, the \emph{bias}
of a configuration is defined as the difference between the number of agents supporting the majority opinion (for short, we   name this number as   \emph{majority}) and the number of agents supporting the second-largest opinion.

Different opinion dynamics have  been studied  in a variety of settings 
\cite{cooper_fast_2017,emanuele_natale_computational_2017}, 
and then used as subroutine to solve more complex computational tasks 
\cite{cruciani_phase_2018,shimizu_phase_2019,boczkowski_minimizing_2018}.

In the aforementioned applicative scenarios, it has been nevertheless observed that only weaker forms of \emph{metastable} consensus are achieved, in which the large majority of agents rapidly achieves a consensus  (while other opinions continue to be supported by  a small set of agents), and this setting is preserved for a relatively-long regime. 
Models that have been considered to study such phenomenon include MAS where: i) agents follow a linear dynamics, such as the \voter model or the  \averaging dynamics and
ii)  a small set of  \emph{stubborn agents} are present in the system \cite{mobilia_does_2003,mobilia_role_2007,yildiz_binary_2013}, or the 
local interactions are affected by \emph{communication noise} \cite{lin_robust_2007}.

We emphasize that the \voter model has a slow (i.e. polynomial in the number $n$ of agents) convergence time even in a fully-connected network (i.e. in the complete graph) and it does not guarantee a high probability to reach consensus on  the initial majority opinion, even starting from a large   initial bias (i.e.   $\Theta(n)$, where $n$ is the number of the agents of the system) \cite{HP01}. On the other hand, averaging dynamics requires agents to perform numerical operations and, very importantly, to have a large local memory (to guarantee a good-enough approximation of real numbers). For the  reasons above,  linear opinion dynamics cannot explain 
fast and reliable metastable consensus  phenomena observed in some MAS \cite{boczkowski_minimizing_2018,feinerman_breathe_2017,condon_approximate_2019}. 

The above discussion naturally leads us to investigate the behaviour of other, non-linear dynamics in the presence of stubborn agents and/or communication noise. Over a MAS having the $n$-node complete graph as the underlying graph, we  introduce a simple model of  \emph{communication noise} in the stochastic  process yielded by   a popular dynamics, known as the \undecided dynamics. In some previous papers~\cite{perron_using_2009}, this protocol has been  called the
\emph{Third-State Dynamics}. We here prefer the term ``undecided''  since it well captures the role of  this additional state.

According to this simple dynamics,      the state of every  agent can be either an opinion (chosen from a finite set $\Sigma$) or  the \emph{undecided state}.
At every discrete-time step (i.e., round),  
every agent ``pulls'' the state of a random neighbor and updates its state according to the following  rule:
if a non-undecided agent pulls a different opinion from its current one, then it will get undecided, while in all other cases it keeps its opinion; moreover, if the node is  undecided  then it will get the state of the pulled neighbor.

This non-linear   dynamics   is known to compute Consensus (and Majority Consensus) on the complete network within a logarithmic number of rounds \cite{angluin_simple_2007,clementi_consensus_2018} and, very importantly, it is optimal in terms of local memory since it requires just one extra  state/opinion \cite{mertzios_determining_2016}.

While communication noise is a common feature of real-world systems and its effects have been 
thoroughly investigated in physics and information theory \cite{cover_elements_2006}, 
its study has been mostly focused on settings in which communication happens over \emph{stable links} 
where the use of error-correcting codes is feasible since message of large size are allowed;
it has been otherwise noted that when interactions among the agents are 
random and opportunistic and consists of very-short messages, 
classical information-theoretic arguments do not carry on and new phenomena calls for a
theoretical understanding \cite{boczkowski_limits_2018}.

\subsubsection*{Our Contribution.} \label{ssec:contrib}
In this work, we   show that, under a simple model of uniform noise, the \undecided dynamics exhibits 
an interesting \emph{phase transition}. 

   We consider the     binary case (i.e., $|\Sigma| =2$) together with  an \emph{oblivious} and \emph{symmetric} action of noise over messages:  
 any sent message is changed  upon being received to any other value, independently and uniformly at random with probability $p$ (where $p$ is any fixed positive
constant smaller than $1/2$).

On one hand, if $p < 1/6$, starting from an arbitrary configuration of the complete network of $n$ agents,
 we prove that the system \emph{with high probability}\footnote{An event $E$ holds \emph{with high probability} if a constant $\gamma > 0$ exists such that $\mathbf{P}(E) \geq 1 - (1/n)^{\gamma}$.}
 (\emph{w.h.p.}, for short)  reaches, within $\bigo(\log n)$ rounds, a metastable almost consensus regime where  the bias towards one fixed valid opinion
keeps  large, i.e. $\Theta(n)$, for at least a $\poly(n)$ number of rounds (see Theorem \ref{thm:symbre}).
In particular, despite the presence of   random communication noise, our result implies that the \undecided dynamics  is able to rapidly break the initial symmetry of any balanced configuration and  reach a metastable regime of almost consensus (e.g., the perfectly-balanced configuration with $n/2$ agents having one opinion and the other $n/2$ agents having the other opinion).

  Importantly enough, our  probabilistic analysis  also shows that,  for any $p <1/6$, the system is able to ``compute'' the task of almost Majority Consensus.  Indeed, in Theorem \ref{theorem_almost_plurality}, 
starting from an arbitrary configuration
with     bias $\Omega(\sqrt{n\log n})$,\footnote{We remark that, when every agent chooses its initial binary opinion uniformly at random, the standard deviation of the bias is $\Theta(\sqrt{n})$.} we prove that the system w.h.p. reaches, within $\bigo(\log n)$ rounds, a metastable regime where the bias towards the  initial majority opinion  keeps large, i.e. $\Theta(n)$,  for at least a $\poly(n)$ number of rounds (see Theorem \ref{theorem_almost_plurality}). For instance, our analysis for $p = 1/10$ implies that the process rapidly reaches a metastable regime where the 
bias keeps size larger than $n/3$.

On the other hand, if $p > 1/6$, even when the initial bias is maximum (i.e., when the system starts from any full-consensus configuration),  after a logarithmic number of rounds,  the information about the initial majority opinion is  ``lost'': in Theorem \ref{theorem_victory_of_noise}, we indeed show that  the system w.h.p.  enters into a regime  where the bias keeps bounded  by  $\bigo(\sqrt{n\log n})$. We also performed some computer simulations that confirm our theoretical results, showing  that  the majority opinion  switches continuously during this regime (see Section \ref{sec:exp} for further details).

Interestingly, in Subsection \ref{ssec:oblivious_noise} 
we  show that our noise model is equivalent to a noiseless setting in which \emph{stubborn} agents are present in the system \cite{yildiz_binary_2013} (that is, agents that never change their state): we thus obtain an analogous phase transition in this setting. 
The obtained phase transition thus  separates qualitatively the behavior of the \undecided dynamics from that of the \voter model which is, to the best of our knowledge, the only opinion dynamics (with a finite opinion set) which has been rigorously analyzed in the presence of communication noise or stubborn agents \cite{mobilia_role_2007,yildiz_binary_2013}: this  hints at a more general phenomenon for dynamics with fast convergence to some metastable consensus.


We believe this work contributes to the research endeavour of exploring the interplay between
communication noise and stochastic interaction pattern in MAS. 
As we will discuss in the Related Work, 
despite the fact that these two characteristics are quite common in real-world MAS, their combined effect is still far from being understood and poses novel mathematical challenges. 
Within such framework,   we have identified and rigorously analyzed a phase transition 
behaviour of the famous \uproc in the presence of communication noise (or, of stubborn agents) on the complete graph. 

\subsubsection*{Related Work.} \label{ssec:related}

The \undecided dynamics has been originally studied as an efficient majority-consensus protocol by 
\cite{angluin_simple_2007} and independently by \cite{benezit_interval_2009} for the binary case
(i.e. with two initial input values). 
They proved that w.h.p., within a logarithmic
number of rounds, all agents support the initial majority opinion. 
Some works have then extended the analysis of the \undecided dynamics to non-complete topologies. 
In the Poisson-clock model (formally equivalent to the Population Protocol model), \cite{draief_convergence_2012}
derive an upper bound on the expected convergence time of the dynamics that holds for arbitrary connected graphs, 
which is based on the location of eigenvalues of some contact rate matrices. 
They also instantiate their bound for particular network topologies. 
Successively, \cite{mertzios_determining_2016} provided an analysis when  the initial states of agents 
are assigned independently at random, and they also derive ``bad'' initial configurations on certain graph topologies such that the initial minority opinion eventually becomes the majority one. 
As for the use of \undecided as a generic consensus protocol,
\cite{DBLP:conf/mfcs/ClementiGGNPS18} recently proved that, 
in the synchronous uniform PULL model in which all agents update their state in parallel by observing the 
state of a random other node the convergence time of the \undecided dynamics is  w.h.p. logarithmic.

The motivation to investigate opinion dynamics is twofold: 
they can be regarded as simplistic models of several real-world phenomena or as building blocks for more complex algorithms. 
While on the modelling side the \undecided dynamics is an appealing model of opinion dynamics and it has 
also been considered as a  model of some mechanism occurring in the biology of a cell \cite{cardelli_cell_2012},
it has been employed as a sub-routine of efficient \emph{Majority Consensus} protocols:
\cite{ghaffari_polylogarithmic_2016}, \cite{berenbrink_efficient_2016} and 
\cite{elsasser_brief_2017} consider  Majority Consensus   in the Uniform-PULL, 
and design protocols (based on the \undecided dynamics) which w.h.p. converge in poly-logarithmic time even if the number 
of initial opinions   is very  large.

Notably, communication noise in random-interacting MAS appears to be a neglected area
of investigation \cite{lin_robust_2007,mas_complex_2016,mas_engineering_2018,CHK18}. 
Such shortage of studies contrasts with the
vast literature on communication noise over \emph{stable}  networks\footnote{For stable networks, we here mean a network
where  communication between agents can be modeled as a classical \emph{channel}  
the agents can use to exchange messages at will \cite{cover_elements_2006}.}. 
Among the few investigations of communication noise in MAS, we note the Vicsek model \cite{Vicsek95}, 
where  agents are driven with a constant absolute
velocity, and at each time step assume the average direction of motion of the agents
in their neighborhood: this strategy is then combined with  some random perturbation. The authors show that the average velocity of their model exhibits a phase transition around 
some critical value of the model parameters which include the noise. 

More recently, in \cite{feinerman_breathe_2017},
the authors consider a settings in which agents interact uniformly at
random by exchanged binary messages which are subject to noise.
In detail, the authors provide simple and efficient protocols to solve the classical 
distributed-computing problems of Broadcast (a.k.a Rumor Spreading) and Majority Consensus, 
in the \PUSH model with binary messages, in which each message can be changed upon being received
with probability $1/2-\epsilon$. 
Their results have been generalized to the Majority Consensus Problem for the multi-valued case
in \cite{fraigniaud_noisy_2018}. 
When the noise is constant, \cite{feinerman_breathe_2017} proves that in their noisy version of the \PUSH model, 
the Broadcast Problem can be solved in \emph{logarithmic} time.
Rather surprisingly, \cite{boczkowski_limits_2018} and \cite{clementi_consensus_2018} 
prove that solving the Broadcast Problem in the \PULL model takes linear time, 
while the time to perform Majority Consensus remains logarithmic in both models.

The fact that real-world systems such as social networks fail to converge to consensus has been extensively studied
in various disciplines; formal models developed to investigate the phenomenon include the multiple-state Axelrod model
\cite{axelrod_dissemination_1997} and the bounded-compromise model by Weisbuch et al \cite{weisbuch_meet_2002}; 
the failure to reach consensus in these models is due to the absence of interaction among agent opinions which are 
\emph{``too far apart''}. 
A different perspective is offered by models which investigate the effect of stubborn agents 
(also known as \emph{zealotry} in the literature), in which some \emph{stubborn/zealot} agents
never update their opinion. 

Several works have been devoted to study such effect under linear models of opinion dynamics.
Starting with \cite{mobilia_does_2003} 
which proposed a statistical-physics method in order to study the \voter model 
under the presence of a stubborn agent, followed by \cite{mobilia_role_2007} which considers 
the case of several stubborn agents in the system. 
Later investigations analyzed various aspects of the stationary distribution of the systems,
such as \cite{acemoglu_opinion_2012,AFF19} which investigate the relationship between
the behavior of the opinion dynamics and the structure of the underlying
interaction graph, 
or \cite{yildiz_binary_2013}, in which the authors consider the \voter dynamics 
and study the first and second moments of the number of the average agents' opinion. 

\smallskip
\noindent 
\textbf{Roadmap of the paper.}
In Section \ref{sec:prel}, we give some preliminaries and the equivalence result between   communication noise and stubborn agents. In Section \ref{sec:phasetrans}, we provide the probabilistic analysis of the \uproc when the initial bias is relatively large and its consequences on almost Majority Consensus. This analysis will be then combined with the analysis of the symmetry-breaking phase given in Section \ref{sec:symbreak} to obtain our results on almost Consensus. Some computer simulations validating experimentally our theoretical results are shown in Section \ref{sec:exp}.

\section{Preliminaries}\label{sec:prel}

We study the discrete-time, parallel version of the  \undecided dynamics on the complete graph in the binary setting \cite{DBLP:conf/mfcs/ClementiGGNPS18}. In detail, there is an additional state/opinion, i.e.\ the \emph{undecided state}, besides  the  two possible opinions (say, opinion \mesalpha and opinion \mesbeta) an agent can support, and, in the absence of noise, the updating rule works as follows: at every round $t\ge 0, t\in \mathbb{N}$, each agent $u$ chooses a neighbor $v$ (or, possibly, itself) independently and uniformly at random and, at the next round,   it 
gets  a new opinion  according to the rule given in  Table \ref{table1}.\footnote{Notice that this dynamics requires no labeling of the agents, i.e., the network can be anonymous.}
	\begin{table}[htb]
		\centering
		\begin{tabular}{c|c c c}
			$u\setminus v$ & undecided & \mesalpha & \mesbeta \\ \hline 
			undecided      & undecided & \mesalpha       & \mesbeta       \\ 
			\mesalpha      & \mesalpha       & \mesalpha       & undecided \\ 
			\mesbeta      & \mesbeta       & undecided & \mesbeta       \\ 
		\end{tabular}
		\caption{The update rule of the USD.}
		\label{table1}
	\end{table}
The definition of noise we consider is the following.

\begin{definition}[Definition of noise]\label{def:noise}
Let $p$ be a real number in the interval $\left(0,1/2\right]$. When an agent $u$ chooses a neighbor $v$ and looks at (pulls) its opinion, it sees $v$'s opinion with probability $1-2p$, and,  with probability $p$, it sees one of the two other opinions.
\end{definition}
For instance, if $v$ supports opinion \mesalpha, then $u$ sees \mesalpha with probability $1-2p$, it sees \mesbeta with probability $p$, and it sees the undecided state with probability $p$. In this work, the terms \emph{agent} and \emph{node} are interchangeable.

\subsubsection{Notation, Characterization, and Expected Values.}

Let us name $C$ the set of all possible configurations; notice that, since the graph is complete and its nodes are anonymous, a configuration $\mathbf{x}\in C$ is uniquely determined by giving the number of \emph{Alpha}  nodes, $a(\mathbf{x})$ and the number of \emph{Beta}  nodes, $b(\mathbf{x})$. Accordingly to this notation, we call $q(\mathbf{x})$ the number of undecided nodes in configuration $\mathbf{x}$, and $s(\mathbf{x})=a(\mathbf{x}) - b(\mathbf{x})$ the \emph{bias} of the configuration $\mathbf{x}$. When the configuration is clear from the context, we will omit $\mathbf{x}$ and write just $a,b,q$, and $s$ instead of $a(\mathbf{x}),b(\mathbf{x}),q(\mathbf{x})$, and $s(\mathbf{x})$. The \udyn defines a finite-state non reversible Markov chain $\{\mathbf{X}_t\}_{t\ge 0}$ with state space $C$ and no absorbing states.

The stochastic process yielded by  the \udyn, starting from a given configuration, will be denoted  as   \uproc. Once a configuration $\mathbf{x}$ at a round $t\ge 0$ is fixed, i.e.\ $\mathbf{X}_t = \mathbf{x}_t$, we use the capital letters A, B, Q, and S to refer to random variables $a(\mathbf{X}_{t+1})$, $b(\mathbf{X}_{t+1})$, $q(\mathbf{X}_{t+1})$, and $s(\mathbf{X}_{t+1})$. Notice that we consider the bias as $a(\config)-b(\config)$ instead of $\abs{a(\config)-b(\config)}$ since the expectation of $\abs{A-B}$ is much more difficult to evaluate than that of $A-B$.

The   expected values of the above key random variables can be written as follows:
\begin{align}
    \mean\left[A\bigm| \mathbf{x} \right] = \ & \frac{a}{n}(a+2q)(1-2p)\nonumber
    \\ & + \left[a(a+b)+(a+q)(b+q)\right]\frac{p}{n}, \label{expectation_A}\\
    \mean\left[B\bigm| \mathbf{x} \right] = \ & \frac{b}{n}(b+2q)(1-2p)\nonumber
    \\ & + \left[b(a+b)+(a+q)(b+q)\right]\frac{p}{n}, \label{expectation_B}\\
    \mean\left[S\bigm| \mathbf{x} \right] = \ & s\left(1-p+(1-3p)\frac{q}{n}\right), \label{expectation_S}\\
    \mean\left[Q\bigm| \mathbf{x} \right] = \ & pn+\frac{1-3p}{2n}\left[2q^2+(n-q)^2-s^2\right]. \label{expectation_Q}
\end{align}
The proof of equations \ref{expectation_S} and \ref{expectation_Q} can be found in Appendix \ref{app:preliminaries}.

\subsubsection{Oblivious Noise and Stubborn Agents.}\label{ssec:oblivious_noise}

We can now consider the following more general 
message-\emph{oblivious} model of noise.
\begin{definition}
    We say that the communication is affected by \emph{oblivious noise} if the value of any sent message
    changes according to the following scheme:
    \begin{itemize}
        \item[(i)] with probability $1-\pnoise$ independent from the value 
            of the sent message, the message remains unchanged;
        \item[(ii)] otherwise, the noise acts on the message and it changes its value according to a 
            fixed distribution $\mathbf{p} = p_1,...,p_m$ over the possible message values $1,...,m$. 
    \end{itemize}
\end{definition}
In other words, according to the previous definition of noise (Definition \ref{def:noise}), the probability that the noise changes any message to message $i$ 
is $\pnoise \cdot p_i$. 
It is immediate to verify that the definition of noise adopted in  Theorems \ref{theorem_almost_plurality} 
and \ref{theorem_victory_of_noise} corresponds to the aforementioned model of oblivious noise 
in the special case $m=3$, $\pnoise = p$, and $p_{\mesalpha}=p_{\mesbeta}=p_{undecided}=\frac 13$. 

Recalling that an agent is said to be \emph{stubborn} if it never updates its state \cite{yildiz_binary_2013}, we now observe that the above noise model   is in fact equivalent to consider   the behavior
of the same dynamics \emph{in a noiseless setting with stubborn agents}. 
\begin{lemma}\label{lem:equiv}
    Consider the \udyn
    on the complete graph with opinions (i.e. message values) in $\Sigma=\{1,...,m\}$. 
    The following two processes are equivalent.
    \begin{itemize}
        \item[(a)] the \uproc with $n$ agents in the presence of oblivious noise with parameters $\pnoise$ and $\mathbf{p} = p_1,...,p_m$;
        \item[(b)] the \uproc with $n$ agents and $\additionalNodes =\frac{\pnoise}{1-\pnoise}n$ additional \emph{stubborn agents} present in the system, of which:
        $\additionalNodes \cdot p_{1}$ are stubborn agents supporting opinion 1, $\additionalNodes \cdot p_{2}$ are stubborn agents supporting opinion 2, and so on. 
    \end{itemize}
\end{lemma}

\begin{proof}[Proof of Lemma \ref{lem:equiv}]
    The equivalence between the two processes is showed through a coupling. Consider the complete graph of $n$ nodes, $K_n$, over which the former process runs. Consider also the complete graph $K_{n+\additionalNodes}$, which contains a sub-graph isomorphic to $K_n$ we denote as $\tilde{K}_n$. Let $H = K_{n+\additionalNodes} \setminus \tilde{K}_n$. The nodes of $H$ are such that $\additionalNodes \cdot p_{1}$ are stubborn agents supporting opinion 1,  $\additionalNodes \cdot p_{\mesbeta}$ are stubborn agents supporting opinion 2, and so on. Observe that $\sum_{i=1}^m p_i = 1$, so this partition of $K_{n+\additionalNodes}$ is well defined. 
    
    The \udyn behaves in exactly the same way over $K_{n+\additionalNodes}$, with the exception that the stubborn agents never change their opinion and that there is no noise perturbing communications between agents. Let $C$ and $\tilde{C}$ be the set of all possible configurations of, respectively, $K_n$ and $K_{n+\additionalNodes}$. Let $\phi: K_n \to \tilde{K}_n$ be any bijective function. The coupling is a bijection $f: C \to \tilde{C}$ such that, for any node $v \in K_n$ in the configuration $\config \in C$, the corresponding node $\phi(v) \in \tilde{K}_n$ in the configuration $f(c) \in \tilde{C}$ supports $v$'s opinion. Consider the two resulting Markov processes $\{\mathbf{X}_t\}_{t\ge 0}$ over $K_n$ and $\{\mathbf{X}_t'\}_{t\ge 0}$ over $K_{n+\additionalNodes}$, denoting the opinion configuration at time $t$ in $K_n$ and in $K_{n+\additionalNodes}$, respectively. It is easy to see that the two transition matrices are exactly the same, namely the probability to go from configuration $c
    \in C$ to configuration $c' \in C$ for $\mathbf{X}_t$ is the same as that to go from configuration $f(c) \in \tilde{C}$ to configuration $f(c') \in \tilde{C}$ for $\mathbf{X}'_t$.  
    
    Indeed, in the former model (a), the probability an agent pulls opinion $j\in\{1,\dots, m\}$ at any given round is
    \[
        (1-\pnoise)\frac{c_j}{n} + \pnoise \cdot p_j\ ,
    \]
    where $c_j$ is the size of the community of agents supporting opinion $j$; in the model defined in (b), the probability a non-stubborn agent pulls opinion $j$ at any given round is 
    \[
        \frac{c_j+\additionalNodes\cdot p_j}{n+\additionalNodes} = \frac{c_j+\frac{\pnoise}{1-\pnoise}n\cdot p_j}{n+\frac{\pnoise}{1-\pnoise}n} = (1-\pnoise)\cdot\frac{c_j}{n}+\pnoise\cdot p_{j}\ .
    \]
\end{proof}
Basically, this equivalence implies that any result we state for the process defined in (a) has an analogous statement for the process defined in (b).

\subsubsection{Probabilistic Tools.}

Our analysis makes use of the following probabilistic result which states that the intersection of some polynomial number of events holding w.h.p.\ is still an event wich holds w.h.p.
\begin{lemma}\label{lemma:whp_intersection}
    Consider any family of events $\{\xi_i\}_{i\in I}$ with $\abs{I} \le n^\lambda$, for some $\lambda > 0$. Suppose that each event $\xi_i$ holds with probability at least $1-n^\eta$, with $\eta > \lambda$. Then, the intersection $\cap_{i\in I}\xi_i$ holds w.h.p.
\end{lemma}
\begin{proof}[Proof of Lemma \ref{lemma:whp_intersection}]
    By the union bound, $\Pr(\cap_{i\in I} \xi_i)=1-\Pr(\cup_{i\in I} \bar\xi_i) \geq 1-\sum_{i\in I} n^{-\eta} = 1-n^{\lambda-\eta} \geq 1- n^{-\delta}$, where $\bar\xi_i$ denotes the negation of $\xi_i$ and $\delta = \frac{\eta - \lambda}{2}$.
\end{proof}

\section{Process analysis for biased initial configurations}\label{sec:phasetrans}

In this section, we analyze the \uproc when the system starts from 
biased configurations. The following two theorems show the phase transition 
exhibited by this process. We remind that our notion of noise is that of Definition \ref{def:noise}.

\begin{theorem}[Almost Majority Consensus]\label{theorem_almost_plurality}
Let $\mathbf{x}$ be any initial configuration having bias $s(\mathbf{x}) \ge \gamma\sqrt{n\log n}$ for some constant $\gamma>0$, and let $\epsilon\in\left(0,1/6\right)$ be some absolute  constant. If $p=1/6 - \epsilon$ is the noise probability, then the \uproc reaches a configuration $\mathbf{y}$ having bias $s(\mathbf{y}) \in \Delta = \left[\frac{2\sqrt{\epsilon}}{1+6\epsilon}n, \left(1-2\left(\frac{1-6\epsilon}{12}\right)^3\right)n\right]$  within $\bigo(\log n)$ rounds, w.h.p. 
Moreover, starting from $\mathbf{y}$, the \uproc enters  a (metastable) phase of length $\Omega\left(n^\lambda\right)$ rounds (for some  constant $\lambda>0$)\footnote{The constant $\lambda$ depends only on the values of $\epsilon$ and $\gamma$. The same holds for the constant $\lambda'$ in Theorem \ref{theorem_victory_of_noise}.} where  the bias  
remains in  the range $\Delta$, w.h.p.
\end{theorem}
Observe that if the theorem is true, then it also holds analogously for the symmetrical case in which $s(\mathbf{x})\le -\gamma\sqrt{n\log n}$.

\begin{theorem}[Victory of Noise]\label{theorem_victory_of_noise}
Let $p=1/6 + \epsilon$ be  the  noise probability for some absolute constant  $\epsilon\in\left(0,1/3\right]$. Assume the system starts from any configuration $\mathbf{x}$ with $\lvert s(\mathbf{x})\rvert \ge \gamma\sqrt{n\log n}$, for some constant $\gamma > 0$. Then, the \uproc reaches a configuration $\mathbf{y}$ having bias $\lvert s(\mathbf{y})\rvert = \bigo(\sqrt{n\log n})$ in $\bigo(\log n)$ rounds, w.h.p. Furthermore, starting from such a configuration, the \uproc enters a (metastable) phase of length $\Omega\left(n^{\lambda'}\right)$ rounds (for some constant $\lambda'>0$) where the absolute value of the bias  keeps bounded by $\bigo(\sqrt{n\log n})$, w.h.p.
\end{theorem}

The next subsections are devoted  to the proof of Theorem \ref{theorem_almost_plurality} (Subsection \ref{ssec:almostplurality}) and Theorem \ref{theorem_victory_of_noise} (Subsection \ref{ssec:victorynoise}). We here just remark that the adopted  arguments in the two proofs are similar. 

Let us now consider the equivalent model with stubborn agents according to Lemma \ref{lem:equiv}, in which $p_{noise} = 3p$ and $p_{\emph{Alpha}}=p_{\emph{Beta}}=p_{undecided}=\frac{1}{3}$. We thus have $n_{stub}=\frac{3p}{1-3p}n$ additional stubborn nodes, of which $n_{stub}\cdot \frac{1}{3}=\frac{p}{1-3p}n$ support opinion \emph{Alpha}, $n_{stub}\cdot \frac{1}{3}=\frac{p}{1-3p}n$ opinion \emph{Beta}, and $n_{stub}\cdot \frac{1}{3}=\frac{p}{1-3p}n$ are undecided. On this new graph of $ n+n_{stub}$ nodes, let the \undecided dynamics run and call the resulting process the \stubproc. The next result is an immediate corollary of the two previous theorems.

\begin{corollary}\label{corollary_theorems}
    Let $\frac 12 >p>0$ be a constant, and let the \stubproc start from any configuration having bias $s \ge \gamma\sqrt{n\log n}$ for some constant $\gamma >0$. 
        If $p < \frac 16$, then, in $\bigo(\log n)$ rounds, the \stubproc enters   a metastable phase of almost consensus of length $\Omega\left(n^\lambda\right)$ for some constant $\lambda > 0$,  in which the bias is $\Theta(n)$, w.h.p.
        If $p \in (\frac 16, \frac 12]$, then, in $\bigo(\log n)$ rounds, the \stubproc enters   a metastable phase of length $\Omega\left(n^{\lambda'}\right)$ for some constant $\lambda' > 0$ where the absolute value of the bias keeps bounded by $\bigo(\sqrt{n\log n})$, w.h.p.
\end{corollary}
Trivially, the corollary holds also in the symmetrical case in which $s\le -\gamma\sqrt{n\log n}$.

\subsection{Proof of Theorem \ref{theorem_almost_plurality}}\label{ssec:almostplurality}

Informally, while the analysis is technically involved, it can be appreciated from it that the phase transition phenomenon at hand  relies ultimately on the exponential drift of the \undecided towards the majority opinion in the absence of noise: 
as long as the noise is kept within a certain threshold, the dynamics manages to quickly amplify and sustain the 
bias towards the majority opinion; as soon as the noise level reaches the threshold, the expected increase of the majority bias 
abruptly decreases below the standard deviation of the process and the ability of the dynamics to preserves a \emph{signal}
towards the initial majority rapidly vanishes.

We now proceed with the formal analysis.
Wlog,  in the sequel, for a given starting configuration $\mathbf{x}$, we will assume $a(\mathbf{x})\ge b(\mathbf{x})$. Indeed, as it will be clear from the results, if $s(\mathbf{x}) \ge \gamma\sqrt{n\log n}$, then the plurality opinion does not change for $\Omega(n^\lambda)$ rounds, w.h.p., and the argument for the case $b(\mathbf{x})>a(\mathbf{x})$ is symmetric.
First notice that, for any fixed  $\epsilon\in(0,1/6)$ and $p=1/6-\epsilon$,      Equations \eqref{expectation_S} and \eqref{expectation_Q} become
\begin{align}
    \mean\left[S\bigm| \mathbf{x}\right] = \ & s\left(\frac{5}{6}+\epsilon+\frac{1}{2}(1+6\epsilon)\frac{q}{n}\right), \label{expectation_S_oknoise}\\
    \mean\left[Q\bigm| \mathbf{x}\right] = \ & \frac{3}{4}\left(\frac{1+6\epsilon}{n}\right)q^2 - \frac{1+6\epsilon}{2}q + \frac{5+6\epsilon}{12}n - \frac{1+6\epsilon}{n}\left(\frac{s}{2}\right)^2. \label{expectation_Q_oknoise}
\end{align} 
The key-point to prove the first claim of the theorem is to show that, if the bias of the configuration is less than $\beta n$ (for some suitable   constant $\beta$), and the number of undecided nodes is some   constant factor of $n$, then the bias at the next round  increases by a constant factor, w.h.p. At the same time, as long as the bias is below $\beta n$, the number of undecided nodes in the next round is sufficiently large, w.h.p.

\begin{lemma}\label{lemma:thm1:biasincrease}
    Let $\mathbf{x}$ be a configuration such that $q\ge \frac{1-4\epsilon}{3(1+6\epsilon)} n$ and $s\ge \gamma\sqrt{n\log n}$ for some constant $\gamma>0$. Then, in the next round, $S\ge s\left(1+\frac{\epsilon}{6}\right)$, w.h.p.
\end{lemma}
\begin{proof}[Proof of Lemma \ref{lemma:thm1:biasincrease}]
    We first notice that  Equation \eqref{expectation_S_oknoise}  implies  
    $\mathbb{E}\left[S\bigm| \mathbf{x} \right] \ge s\left(1+\epsilon/3\right)$.
    Then, consider  the events  
    \[ E_1 = \left\{A \le \mathbb{E}\left[A \mid \mathbf{x}\right] - \frac{\epsilon}{12}\gamma\sqrt{n\log n}\right\} \, \mbox{ and } \, E_2 = \left\{B \ge \mathbb{E}\left[B \mid \mathbf{x}\right] + \frac{\epsilon}{12}\gamma\sqrt{n\log n}\right\} \]
    For the additive form of Chernoff bound (Theorem \ref{chernoff:additive} in Appendix \ref{tools}), it holds that
    \begin{align*}
        \mathbb{P}\left(E_1 \bigm| \mathbf{x}\right) \le e^{-\frac{2n\log n}{144n}}= n^{-\frac{1}{77}} \ \ \text{and} \ \
        \mathbb{P}\left(E_2 \bigm| \mathbf{x}\right) \le e^{-\frac{2n\log n}{144n}}= n^{-\frac{1}{77}}.
    \end{align*}
   It follows that
    \begin{align*}
        & \mathbb{P}\left(S \ge s\left(1+\frac{\epsilon}{6}\right) \mid \mathbf{x} \right) 
        =   \mathbb{P}\left(S \ge s\left(1+\frac{\epsilon}{3}\right) - \frac{\epsilon}{6}s \mid \mathbf{x} \right)  \\
        \ge \ &   \mathbb{P}\left(S \ge \mathbb{E}\left[S \mid \mathbf{x}\right] - \frac{\epsilon}{6}\gamma\sqrt{n\log n} \mid \mathbf{x} \right)  \\
        = \ &   \mathbb{P}\left(A-B \ge \mathbb{E}\left[A-B \mid \mathbf{x}\right] - 2\frac{\epsilon}{12}\gamma\sqrt{n\log n} \mid \mathbf{x} \right) 
        \ge  \mathbb{P}\left(E_1^C \cap E_2^C \bigm| \mathbf{x}\right) \\
        = \ & \mathbb{P}\left(E_1^C\bigm| \mathbf{x} \right) + \mathbb{P}\left(E_2^C\bigm| \mathbf{x}\right) - \mathbb{P}\left(E_1^C\cup E_2^C \bigm| \mathbf{x}\right) 
        \ge  1-2n^{-\frac{1}{77}},
    \end{align*}
    where in the last inequality we bounded the probability of the union with 1.
\end{proof}

We now fix  $\beta = \frac{2\sqrt{3\epsilon}}{1+6\epsilon}$ and show the following bound.
    
\begin{lemma}\label{lemma:thm1:enoughundecided}
Let $\mathbf{x}$ be a configuration such that $s\le \beta n$. Then, in the next round, $Q\ge \frac{1-4\epsilon}{3(1+6\epsilon)}n$, w.h.p.
\end{lemma}
\begin{proof}[Proof of Lemma \ref{lemma:thm1:enoughundecided}]
Since Equation \eqref{expectation_Q_oknoise} has its minimum in $\bar{q}=\frac{n}{3}$, 
 \begin{align*}
    \mathbb{E}\left[Q\bigm| \mathbf{x}\right] \ge \  (1+6\epsilon)\frac{n}{12}-(1+6\epsilon)\frac{n}{6}+(5+6\epsilon) \frac{n}{12} - (1+6\epsilon)\left(\frac{\beta}{2}\right)^2n  \\
    = \  \frac{n}{12}\left(1+6\epsilon-2-12\epsilon+5+6\epsilon-\frac{36\epsilon}{1+6\epsilon}\right)  
    = \  \frac{1-3\epsilon}{3(1+6\epsilon)}n.
\end{align*}
Hence,  we can apply the additive form of Chernoff bound (Theorem \ref{chernoff:additive} in Appendix \ref{tools}), and get  $Q\ge \frac{1-4\epsilon}{3(1+6\epsilon)}n$, w.h.p.\ \big(actually, with probability $1-\exp(\Theta(n))$\big). Formally,
\begin{align*}
    & \mathbb{P}\left(Q \le \frac{1-4\epsilon}{3(1+6\epsilon)}n \mid \mathbf{x}\right)
    =  \mathbb{P}\left(Q \le \frac{1-3\epsilon}{3(1+6\epsilon)}n - \frac{\epsilon}{3(1+6\epsilon)}n \mid \mathbf{x}\right) \\
    \le \ & \mathbb{P}\left(Q \le \mathbb{E}\left[Q \mid \mathbf{x}\right] - \frac{\epsilon}{3(1+6\epsilon)}n \mid \mathbf{x}\right) 
    \le \  e^{-\frac{2}{n}\frac{\epsilon^2}{9(1+6\epsilon)^2}n^2} = e^{-\frac{2\epsilon^2}{9(1+6\epsilon)^2}n}. \qedhere
\end{align*}
\end{proof}
   The two lemmas above ensure that the system eventually reaches a configuration $\mathbf{y}$ with bias $s(\mathbf{y}) > \beta n$ within $\bigo(\log n)$ rounds, w.h.p. (see the proof of Theorem \ref{theorem_almost_plurality}).
We now consider configurations in which $s>\beta n$ and  derive a useful bound on the possible decrease of $s$.

\begin{lemma}\label{lemma:thm1:biasdecrease}
Let $\mathbf{x}$ be any configuration such that $s\ge \gamma\sqrt{n\log n}$ for some constant $\gamma>0$. Then, in the next round, it holds that $S\ge s\left(\frac{5}{6}+\frac{\epsilon}{2}\right)$ w.h.p.
\end{lemma}
\begin{proof}[Proof of Lemma \ref{lemma:thm1:biasdecrease}]
    Observe   that  Equation \eqref{expectation_S_oknoise} implies
    $\mathbb{E}\left[S\bigm| \mathbf{x} \right] \ge s\left( 5/6+\epsilon\right)$.
    By the additive form of Chernoff bound and the union bound (as we did in the proof of Lemma \ref{lemma:thm1:biasincrease}), we get  $S\ge s\left(\frac{5}{6}+\frac{\epsilon}{2}\right)$, w.h.p.
\end{proof}

Lemma \ref{lemma:thm1:biasdecrease} is used to show the metastable phase of almost consensus, which lasts for a polynomial number of rounds and in which the bias keeps lower bounded by $\frac{2\sqrt{\epsilon}}{1+6\epsilon}n$ (see the proof of Theorem \ref{theorem_almost_plurality}). The next two lemmas provide an upper bound on the bias during this phase.

\begin{lemma}\label{lemma:thm1:lowboundundecided}
    Let $\mathbf{x}$ be any configuration. Then, in the next round, $Q\ge \frac{n}{12}(1-6\epsilon)$, w.h.p.
\end{lemma}
\begin{proof}[Proof of Lemma \ref{lemma:thm1:lowboundundecided}]
From Equation \eqref{expectation_Q_oknoise}
\begin{align*}
        \mathbb{E}\left[Q\bigm| \mathbf{x}\right] \ge \ & \frac{3}{4}\left(\frac{1+6\epsilon}{n}\right)q^2 - \frac{1+6\epsilon}{2}q + \frac{5+6\epsilon}{12}n  - \frac{1+6\epsilon}{n}\left(\frac{n-q}{2}\right)^2  \\
        \ge \ & \frac{1}{2}\left(\frac{1+6\epsilon}{n}\right)q^2 + \frac{1-6\epsilon}{6}n \ge \  \frac{1-6\epsilon}{6}n ,
\end{align*}
where we used  $s\le n-q$. For the additive form of Chernoff bound (Theorem \ref{chernoff:additive} in Appendix \ref{tools}), we get  $Q\ge \frac{1-6\epsilon}{12}n$, w.h.p.
\end{proof}


\begin{lemma}\label{lemma:thm1:lowboundminopinion}
    Let $\mathbf{x}$ be a configuration with $q\ge \frac{n}{12}(1-6\epsilon)$. Then, in the next round, $B\ge \left(\frac{1-6\epsilon}{12}\right)^3n$, w.h.p.
\end{lemma}
\begin{proof}[Proof of Lemma \ref{lemma:thm1:lowboundminopinion}]
From the last term of Equation \eqref{expectation_B}, we have 
\[    
    \mathbb{E}\left[B\bigm| \mathbf{x}\right] \ge \frac{1}{6}\left(\frac{1-6\epsilon}{n}\right)\left(q^2\right) \ge \frac{n}{6\cdot 12^2}(1-6\epsilon)^3.
\]
The the additive form of Chernoff bound (Theorem \ref{chernoff:additive} in Appendix \ref{tools}) implies that $B \ge \left(\frac{1-6\epsilon}{12}\right)^3n$, w.h.p.
\end{proof}

\begin{proof}[Proof of Theorem \ref{theorem_almost_plurality}]

Let $\config$ be the initial configuration. We now  prove that the bias keeps upper bounded by the value $\left(1-2[(1-6\epsilon)/12]^3\right)n$. Indeed, Lemma \ref{lemma:thm1:lowboundundecided}  ensures that the number of undecided nodes keeps at least $\frac{n}{12}(1-6\epsilon)$, w.h.p. Thus, applying Lemmas \ref{lemma:whp_intersection} and \ref{lemma:thm1:lowboundminopinion}, we get that $b(\mathbf{X}_t) \ge [(1-6\epsilon)/12]^3n$, w.h.p., for a polynomial number of rounds .

As for the lower bound of the bias, we distinguish two initial cases.
    
    \noindent
    \emph{Case} $s(\config) \ge \beta n$. From Lemma \ref{lemma:thm1:biasdecrease}, we know that as long as the bias is of magnitude $\Omega(\sqrt{n\log n})$, then it cannot decrease too fast w.h.p., namely $s(\mathbf{X}_{t+1}) \ge s(\mathbf{X}_t)(5/6 + \epsilon/2)$, w.h.p. Notice that 
    \[
        \left(\frac{5}{6}+\frac{\epsilon}{2}\right)^2\cdot \beta n \ge \frac{2\sqrt{\epsilon}}{1+6\epsilon}n,
    \]
    which means that, if at some round $t$ the bias goes below the value $\beta n$, then it remains at least $\frac{2\sqrt{\epsilon}}{1+6\epsilon}n$ and it will not decrease below that value for at least another round, w.h.p. Then, by Lemma \ref{lemma:thm1:enoughundecided} we know that at round $t+1$ the number of undecided nodes is at least $\frac{1-4\epsilon}{3(1+6\epsilon)}n$, w.h.p., which means that the bias starts increasing again each round due to Lemma \ref{lemma:thm1:biasincrease}, w.h.p., as long as it is still below $\beta n$. Indeed, the number of undecided nodes keeps greater than $\frac{1-4\epsilon}{3(1+6\epsilon)}n$ as long as the bias is below $\beta n$, w.h.p. (Lemma \ref{lemma:thm1:enoughundecided}). This phase, in which the bias keeps greater than $\frac{2\sqrt{\epsilon}}{1+6\epsilon}n$, lasts for a polynomial number of rounds, w.h.p.\ (see Lemma \ref{lemma:whp_intersection});
    
    \noindent
    \emph{Case} $\gamma\sqrt{n\log n} \le s(\config) < \beta n$. Thanks to   Lemma \ref{lemma:thm1:biasdecrease}, in the next round, the bias is greater than $\gamma'\sqrt{n\log n}$, w.h.p.,  while the number of undecided nodes gets greater than $\frac{1-4\epsilon}{3(1+6\epsilon)}n$, w.h.p. (Lemma \ref{lemma:thm1:enoughundecided}). Then, Lemmas \ref{lemma:thm1:biasincrease} and \ref{lemma:thm1:enoughundecided} guarantee that, within the next $\bigo\left(\log n\right)$ rounds, the bias reaches the value $\beta n$, w.h.p. (Lemma \ref{lemma:whp_intersection}), and so the process turns to be  in the first Case.

We finally remark that our analysis above shows that  the polynomial length of  the metastable phase, i.e. $n^\lambda$, has the exponent  $\lambda$ that  (only) depends   on the (constant) parameters $\gamma$ and $\epsilon$ of the considered process.
\end{proof}

\subsection{Proof of Theorem \ref{theorem_victory_of_noise}}\label{ssec:victorynoise}

First we present all the necessary technical lemmas (with their proof) we are going to use to prove the theorem and then we prove the theorem. We assume the starting configuration $\mathbf{x}$ to have bias $s(\mathbf{x}) = a(\mathbf{x}) - b(\mathbf{x}) \ge \gamma\sqrt{n \log n}$ for some constant $\gamma >0$; the case in which $b(\mathbf{x})>a(\mathbf{x})$ is analogous. Let $\epsilon\in\left(0,\frac{1}{3}\right]$ be a constant, and $p=1/6+\epsilon$ be the probability of noise. 
Equations \eqref{expectation_S} and \eqref{expectation_Q} become
\begin{align}
    \mean\left[S\bigm| \mathbf{x}\right] = \ & s\left(\frac{5}{6}-\epsilon+\frac{1}{2}(1-6\epsilon)\frac{q}{n}\right), \label{expectation_S_toomuchnoise}\\
    \mean\left[Q\bigm| \mathbf{x}\right] = \ & \frac{3}{4}\left(\frac{1-6\epsilon}{n}\right)q^2 - \frac{1-6\epsilon}{2}q + \frac{5-6\epsilon}{12}n \nonumber
    \\
    & - \frac{1-6\epsilon}{n}\left(\frac{s}{2}\right)^2. \label{expectation_Q_toomuchnoise}
\end{align}
From Equation \eqref{expectation_S_toomuchnoise} it is clear that the bias decreases in expectation exponentially fast each round as long as $\epsilon\ge 1/6$ (actually, $\epsilon > 1/12$ is enough) or $q/n <1/3 \cdot (1+6\epsilon)/(1-6\epsilon)$. We analyze two cases: $\epsilon>1/12$ and $\epsilon\le1/12$.

\subsubsection{First Case: \texorpdfstring{$\epsilon > \frac{1}{12}$}{firstcase} {Large Epsilon}.}

We first show a bound  on the decrease of the bias.

\begin{lemma}\label{lemma:thm2:bigepsilon}
    Let $\mathbf{x}$ be a configuration such that $s\ge \gamma\sqrt{n\log n}$ for some constant $\gamma>0$. Then, in the next round, it holds that $S \le s\left(1-2\epsilon+\frac{1}{6}\right)$ w.h.p.
\end{lemma}
\begin{proof}[Proof of Lemma \ref{lemma:thm2:bigepsilon}]
    From Equation \eqref{expectation_S_toomuchnoise}, we have
    \begin{align*}
        \mathbb{E}\left[S\bigm | \mathbf{x}\right] & \le s\left(\frac{5}{6}-\epsilon+\frac{1}{2}(1-6\epsilon)\right)=s\left(\frac{4}{3}-4\epsilon\right) \\ & \le s\left(1-4\epsilon+\frac{1}{3}\right).
    \end{align*}
    Observe that $4\epsilon-\frac{1}{3}>0$ if and only if $\epsilon>\frac{1}{12}$. Now, let $\lambda = \gamma\left(\epsilon - \frac{1}{12}\right)\sqrt{n\log n}$, and define the events
    \begin{align*}
        E_1  & = \{A \ge \expect{A \mid \config} + \lambda\}, \ \ \text{and} \\
        E_2 & = \{B \le \expect{B \mid \config} - \lambda\}.
    \end{align*}
    Then, for the additive form of Chernoff bound (Theorem \ref{chernoff:additive} in Appendix \ref{tools}) it holds that
    \begin{align*}
        \pr{E_1 \mid \config} \le & \  e^{-\frac{\gamma(\epsilon-1/12)^2 \log n}{2}} = n^{-\frac{\gamma(\epsilon-1/12)^2}{2}}, \ \ \text{and} \\
        \pr{E_2 \mid \config} \le & \  e^{-\frac{\gamma(\epsilon-1/12)^2 \log n}{2}} = n^{-\frac{\gamma(\epsilon-1/12)^2}{2}}.
    \end{align*}
    Then, for the union bound, we have that
    \begin{align*}
        \pr{S \le s\left(1-2\epsilon+\frac{1}{6}\right) \mid \config } = & \ \pr{S \le s\left(1-4\epsilon+\frac{1}{3}\right) + 2\lambda \mid \config } \\
        \ge & \ \pr{S \le \expect{S \mid \config} + 2\lambda \mid \config } \\
        = & \ \pr{A-B \le \expect{A-B \mid \config} + 2\lambda \mid \config } \\
        \ge & \ \pr{E_1^C, E_2^C \mid \config} \\
        \ge & \  \pr{E_1^C \mid \config} + \pr{E_2^C \mid \config} - \pr{E_1^C \cup E_2^C \mid \config} \\
        \ge & \ 1 - 2n^{-\frac{\gamma(\epsilon-1/12)^2}{2}},
    \end{align*}
    where in the last inequality we used that the probability of the union of two events is at most 1.
\end{proof}

We then obtain the first part of Theorem \ref{theorem_victory_of_noise} by using a simple symmetric argument for configurations such that $b>a$, and by Lemma \ref{lemma:whp_intersection}.

\subsubsection{Second Case: \texorpdfstring{$\epsilon\in \left(0,\frac{1}{12}\right]$}{secondcase} {Small Epsilon}.}\label{ssec:proof_noise_secondcase}

The following lemma states that, if $s\ge \frac{2}{3}n$, the bias decreases exponentially at the next round, w.h.p. On the other hand, if the bias is at most $\frac{2}{3}n$, it cannot grow over $\frac{2}{3}n$, w.h.p.

\begin{lemma}\label{lemma:thm2:biascontrol}
    Let $\mathbf{x}$ be any configuration. The followings hold:
    \begin{itemize}
        \item[(1)] if $s\ge \frac{2}{3}n$, then $S \le s(1-\epsilon) $ w.h.p.;
        \item[(2)] if $s\le \frac{2}{3}n$, then $S\le \frac{2}{3}n$ w.h.p.
    \end{itemize}
\end{lemma}
\begin{proof}[Proof of Lemma \ref{lemma:thm2:biascontrol}]
    Consider the first statement. If $s\ge \frac{2}{3}n$, then $q< \frac{1}{3}n$. Thus
    \[
        \mathbb{E}\left[S\bigm| \mathbf{x}\right] \le s\left(\frac{5}{6}-\epsilon+\frac{1}{6}-\epsilon\right) \le s(1-2\epsilon).
    \]
    We conclude using the additive form of Chernoff bound (Theorem \ref{chernoff:additive} in Appendix \ref{tools}), as we have done in the proof of Lemma \ref{lemma:thm2:bigepsilon}, getting that $S\le s(1-\epsilon)$, w.h.p.
    
    As for the second statement, we take the expectation of $S$ from Equation \eqref{expectation_S_toomuchnoise} and we observe that
    \begin{align*}
            \mathbb{E}\left[S\bigm| \mathbf{x}\right] & = s\left(\frac{5}{6}-\epsilon+\frac{1}{2}(1-6\epsilon)\frac{q}{n}\right)  \\
            & \le s\left(\frac{5}{6}-\epsilon+\frac{1}{2}(1-6\epsilon)\frac{n-s}{n}\right)  \\
            & = s\left(\frac{4}{3}-4\epsilon-\frac{s}{2n}(1-6\epsilon)\right) \\ 
            & \le \frac{2}{3}n\left(1-2\epsilon\right).
        \end{align*}
    We conclude applying the additive form of Chernoff bound (Theorem \ref{chernoff:additive} in Appendix \ref{tools}) on $A$ and $B$, and the union bound, as we have done in the proof of Lemma \ref{lemma:thm2:bigepsilon}, getting that $S\le \frac{2}{3}n$, w.h.p.
\end{proof}
Thus, we just have to take care of cases in which the bias is no more than $\frac{2}{3}n$. The key-point to show the decrease of the bias, as long as it is $\Omega\left(\sqrt{n\log n}\right)$, is the condition $q\le\frac{n(1+3\epsilon)}{3(1-6\epsilon)}$, as shown in the next lemma.

\begin{lemma}\label{lemma2_3}
    Let $\mathbf{x}$ be a configuration such that $s\ge \gamma\sqrt{n\log n}$ for some constant $\gamma >0$. If $q\le  \frac{n}{3}\left(\frac{1+3\epsilon}{1-6\epsilon}\right)$, then in the next round it holds that $S\le s\left(1-\frac{\epsilon}{4}\right)$ w.h.p.
\end{lemma}
\begin{proof}[Proof of Lemma \ref{lemma2_3}]
    From Equation \eqref{expectation_S_toomuchnoise} it follows that
    \[\mathbb{E}\left[S\bigm| \mathbf{x}\right] \le s\left(\frac{5}{6}-\epsilon+\frac{1}{6}+\frac{\epsilon}{2}\right) \le s\left(1-\frac{\epsilon}{2}\right).\]
    The thesis follows from an easy application of the additive form of Chernoff bound (Theorem \ref{chernoff:additive} in Appendix \ref{tools})  and the union bound, as we have done in the proof of Lemma \ref{lemma:thm2:bigepsilon}, getting that $S\le s(1-\epsilon/4)$, w.h.p.
\end{proof}

We now analyze the dynamics by partitioning the interval $\left(0, \frac{2}{3}n\right]$ and seeing what happens to the bias in each element of the partition.
Let $\beta=\frac{2\sqrt{2\epsilon}}{\sqrt{(1+6\epsilon)(1-6\epsilon)}}$ and define $S_{-1} \coloneqq \left(0, \beta n \right]$, $S_i$ the sequence of intervals \[S_i \coloneqq \left(\left(\frac{3}{2}\right)^{i}\beta n, \left(\frac{3}{2}\right)^{i+1}\beta n\right] \] for $i=0, 1, \dots, k-2$ where $k=\left \lceil{\log_\frac{2}{3}(\beta )-1}\right \rceil $, and $S_{k-1}\coloneqq \left(\left(\frac{3}{2}\right)^{k-1}\beta n, \frac{2}{3} n\right]$. Furthermore, just for completeness, we define $S_k \coloneqq \left(\frac{2}{3}n, n\right]$. In the next lemmas, we show that as long as $s\in S_i$ for $i=-1, \dots, k-1$, $s = \Omega(\sqrt{n\log n})$, and $\bar{q}_{i+1} \le q \le  \frac{n}{3}\left(\frac{1+3\epsilon}{1-6\epsilon}\right) $ for some decreasing sequence $\bar{q}_{-1}, \dots, \bar{q}_k$ accurately chosen, then, at the next round, the bias decreases exponentially w.h.p. and the number of undecided nodes moves to the interval $\left[\bar{q}_i, \frac{n(1+3\epsilon)}{3(1-6\epsilon)}\right]$ w.h.p. Note that since $\epsilon$ is a constant, so it is $k$. The following lemma determines the sequence $\bar{q}_i$.

\begin{lemma}\label{lemma2_4}
    Let $\mathbf{x}$ be any configuration. 
    \begin{enumerate}
        \item If $-1\le i\le k-1$ and $s \le \left(\frac{3}{2}\right)^{i+1}\beta n$, it holds that $Q \ge \frac{n}{3} - \frac{2n\epsilon}{1+6\epsilon} \left(\frac{3}{2}\right)^{2i+3}$ w.h.p.; \label{item_2_5_1}
        \item If $s\le \frac{2}{3}n$, it holds that $Q \ge\frac{2}{9}n + \frac{\epsilon}{3}n$ w.h.p.; \label{item_2_5_2}
        \item Without any condition on $s$, it holds that $Q \ge \frac{n}{12} + \epsilon n$ w.h.p. \label{item_2_5_3}
    \end{enumerate}
\end{lemma}
\begin{proof}[Proof of Lemma \ref{lemma2_4}]
    We start proving Item \ref{item_2_5_1}. From Equation \eqref{expectation_Q_toomuchnoise}, we have that
    \begin{align*}
        \mathbb{E}\left[Q \bigm| \mathbf{x}\right] & \ge \frac{n}{3} - \frac{1-6\epsilon}{n}\left(\frac{s}{2}\right)^2\\
        &\ge \frac{n}{3} - \frac{1-6\epsilon}{4}\left(\frac{3}{2}\right)^{2i+2}\beta^2n \\
        & = n\left[\frac{1}{3} - \frac{2\epsilon}{1+6\epsilon}\left(\frac{3}{2}^{2i+2}\right)\right].
    \end{align*}
    Thus, using the additive form of Chernoff bound (Theorem \ref{chernoff:additive} in Appendix \ref{tools}) with $\lambda = \frac{\epsilon}{1+6\epsilon}\left(\frac{3}{2}\right)^{2i+2}n$, we have that $Q \ge n\left(\frac{1}{3} - \frac{2\epsilon }{1+6\epsilon}\left(\frac{3}{2}\right)^{2i+3}\right)$, w.h.p.
    
    As for Item \ref{item_2_5_2}, we have that 
    \begin{align*}
        \mathbb{E}\left[Q \bigm| \mathbf{x}\right] & \ge \frac{n}{3} - \frac{1-6\epsilon}{n}\left(\frac{s}{2}\right)^2 \ge \frac{n}{3} - n\frac{1-6\epsilon}{9} \\ 
        & = \frac{2+6\epsilon}{9}n
    \end{align*}
    and we conclude by using the additive Chernoff  bound (Theorem \ref{chernoff:additive} in Appendix \ref{tools}) with $\lambda = \frac{\epsilon}{3}n$, getting that $Q \ge \frac{2}{9}n + \frac{\epsilon}{3}n$, w.h.p.
    
    To prove Item \ref{item_2_5_3} we use that $s\le n$ and observe that 
    \begin{align*}
        \mathbb{E}\left[Q \bigm| \mathbf{x}\right] & \ge \frac{n}{3} - \frac{1-6\epsilon}{n}\left(\frac{s}{2}\right)^2 \ge \frac{n}{3} - n\frac{1-6\epsilon}{4} \\
        & =\frac{1+18\epsilon}{12}n.
    \end{align*}
    We conclude with the additive Chernoff bound (Theorem \ref{chernoff:additive} in Appendix \ref{tools}) with $\lambda = \frac{\epsilon}{2}n$, getting $Q\ge \frac{n}{12}+\epsilon n$, w.h.p.
\end{proof}

Define $\bar{q}_i \coloneqq n\left(\frac{1}{3} - \frac{2\epsilon }{1+6\epsilon}\left(\frac{3}{2}\right)^{2i+3}\right)$ for $i=-1, ..., k-2$, $\bar{q}_{k-1} \coloneqq \frac{2}{9}n + \frac{\epsilon}{3}n$, and $\bar{q}_k \coloneqq \frac{n}{12}+\epsilon n$, and notice that they form a decreasing sequence. Next, with few lemmas, we take care of controlling the behaviour of the number of undecided nodes when $s > \inf(S_i)$ for $-1\le i\le k-1$.

\begin{lemma}\label{lemma2_5}
    Let $-1\le i\le k-1$ and let $\mathbf{x}$ be a configuration such that $\bar{q}_{i+1} \le q \le \frac{n(1+3\epsilon)}{3(1-6\epsilon)}$ and $s > \inf (S_i)$. Then, at the next round, $\bar{q}_i \le Q\le \frac{n(1+3\epsilon)}{3(1-6\epsilon)}$ w.h.p.
\end{lemma}
\begin{proof}[Proof of Lemma \ref{lemma2_5}]
    Define $f(q)$ equal to $
    \mathbb{E}\left[Q\mid \mathbf{x}\right] =  \frac{3}{4}\left(\frac{1-6\epsilon}{n}\right)q^2 - \frac{1-6\epsilon}{2}q + \frac{5-6\epsilon}{12}n - \frac{1-6\epsilon}{n}\left(\frac{s}{2}\right)^2$. We are going to evaluate $f(q)$ in $\bar{q}_{i+1}$ and in $\bar{q} = \frac{n(1+3\epsilon)}{3(1-6\epsilon)}$. We take care of different cases: first, we assume $i=-1$, with the condition that $s> 0$. Thus
    \begin{align*}
           f(\bar{q}_0) \le \ & \frac{3}{4}\left(\frac{1-6\epsilon}{n}\right)n^2 \bigg[\frac{1}{9} + \frac{4\epsilon^2}{(1+6\epsilon)^2}\left(\frac{3}{2}\right)^{6} - \frac{4\epsilon}{3(1+6\epsilon)}\left(\frac{3}{2}\right)^{3}\bigg]\\
           & - \frac{1-6\epsilon}{2}n \bigg[\frac{1}{3} - \frac{2\epsilon}{1+6\epsilon}\left(\frac{3}{2}\right)^{3} \bigg] + \frac{5-6\epsilon}{12}n \\
           = \ & \frac{n}{3} + n\frac{3\epsilon^2(1-6\epsilon)}{(1+6\epsilon)^2}\left(\frac{3}{2}\right)^{6}  \\
           = \ & \frac{n}{3}\left[1+\frac{3^8\epsilon^2(1-6\epsilon)}{2^6(1+6\epsilon)^2}\right];
    \end{align*}
    now, we observe that
    \begin{align*}
            & 1+\frac{3^8\epsilon^2(1-6\epsilon)}{2^6(1+6\epsilon)^2} - \frac{1}{1-6\epsilon}
            \\
            < \ &  \frac{-6\epsilon}{1-6\epsilon}+ \frac{3^8\epsilon^2}{2^6(1+6\epsilon)^2} \\
            = \ & \frac{-2^6\cdot6\epsilon(1+6\epsilon)^2+3^8\epsilon^2(1-6\epsilon)}{2^6(1-6\epsilon)(1+6\epsilon)^2} \\
            = \ & \frac{3\epsilon\left(-2^7-3\cdot2^9\epsilon-3^2\cdot2^9\epsilon^2+3^7\epsilon-2\cdot3^8\epsilon^2\right)}{2^6(1-6\epsilon)(1+6\epsilon)^2} \\
            = \ & \frac{3\epsilon\left(-2^7+3\epsilon(3^6-2^9)-2\cdot3^2\epsilon^2(2^8+3^6)\right)}{2^6(1-6\epsilon)(1+6\epsilon)^2} <0
    \end{align*}
    where in the last inequality we have used that $-2^7+3\epsilon(3^6-2^9)<0$ for $\epsilon \le \frac{1}{12}$. Thus, $f(\bar{q}_0) < \frac{n}{3(1-6\epsilon)}$.
    
    Second, we assume $0\le i\le k-3$, with the condition that $s > \left(\frac{3}{2}\right)^i\beta n$.
    \begin{align*}
           f(\bar{q}_{i+1}) \le \ & \frac{3}{4}\left(\frac{1-6\epsilon}{n}\right)n^2 \bigg[\frac{1}{9} + \frac{4\epsilon^2}{(1+6\epsilon)^2}\left(\frac{3}{2}\right)^{4i+10} - \frac{4\epsilon}{3(1+6\epsilon)}\left(\frac{3}{2}\right)^{2i+5}\bigg]\\
           & - \frac{1-6\epsilon}{2}n \bigg[\frac{1}{3} - \frac{2\epsilon}{1+6\epsilon}\left(\frac{3}{2}\right)^{2i+5} \bigg] + \frac{5-6\epsilon}{12}n - \frac{2\epsilon}{1+6\epsilon}\left(\frac{3}{2}\right)^{2i}n \\
           = \ & \frac{n}{3} + n\frac{3\epsilon^2(1-6\epsilon)}{(1+6\epsilon)^2}\left(\frac{3}{2}\right)^{4i+10} - n\frac{2\epsilon}{1+6\epsilon}\left(\frac{3}{2}\right)^{2i} \\
           = \ & \frac{n}{3}\bigg\{1+\frac{3\epsilon}{(1+6\epsilon^2)}\left(\frac{3}{2}\right)^{2i}\bigg[3\epsilon(1 -6\epsilon)\left(\frac{3}{2}\right)^{2(i+5)} -2(1+6\epsilon) \bigg]\bigg\};
    \end{align*}
    for the evaluation of $f(\bar{q}_{i+1})$ we observe that $\beta$ is a constant in $(0,1)$ and that 
    \begin{align*}
            & 3\epsilon(1-6\epsilon)\left(\frac{3}{2}\right)^{2(i+5)}-2(1+6\epsilon) \\ 
            < \ & 3\epsilon\left(\frac{2}{3}\frac{1}{\beta}\left(\frac{3}{2}\right)^3\right)^2 - 2 \\
            < \ & \frac{3^5}{2^7} - 2 < 0
    \end{align*} 
    because $\left(\frac{3}{2}\right)^{i+2}< \frac{2}{3\beta}$ for $i+2 \le k$; thus $f(\bar{q}_{i+1}) \le \frac{n}{3}$. 
    
    Let now $i=k-2$; thus $s> \left(\frac{3}{2}\right)^{k-2}\beta n$. We now evaluate $f(\bar{q}_{k-1})$.
    \begin{align*}
           f(\bar{q}_{k-1}) \le \ & \frac{3}{4}\left(\frac{1-6\epsilon}{n}\right)n^2 \left[\frac{4}{81}+\frac{\epsilon^2}{9}+\frac{4\epsilon}{27}\right] - \frac{1-6\epsilon}{2}n \left[\frac{2}{9} + \frac{\epsilon}{3}\right] + \frac{5-6\epsilon}{12}n \\
           & - \frac{2\epsilon}{1+6\epsilon}\left(\frac{3}{2}\right)^{2k-4}n \\
           = \ & \frac{-8(1-6\epsilon)+45-54\epsilon}{108}n-\frac{\epsilon(1-6\epsilon)}{18}n  - \frac{2\epsilon}{1+6\epsilon}\left(\frac{3}{2}\right)^{2k-4}n \\
           = \ & \frac{37 - 12\epsilon + 36\epsilon^2}{108}n - \frac{2\epsilon}{1+6\epsilon}\left(\frac{3}{2}\right)^{2k-4}n \\
           = \ & \frac{n}{3}\left[1+\frac{(1-6\epsilon)^2}{36}-\frac{6\epsilon}{1+6\epsilon}\left(\frac{3}{2}\right)^{2k-4}\right] \\
           = \ & \frac{n}{3}\biggl[1+\left(\frac{1-6\epsilon}{6}-\sqrt{\frac{6\epsilon}{1+6\epsilon}}\cdot\frac{3^{k-2}}{2^{k-2}}\right) \cdot \left(\frac{1-6\epsilon}{6}+\sqrt{\frac{6\epsilon}{1+6\epsilon}}\cdot\frac{3^{k-2}}{2^{k-2}}\right)\biggr].
    \end{align*}
    Observe that, by definition of $k$, we have
    \begin{align*}
    & \frac{1-6\epsilon}{6}-\sqrt{\frac{6\epsilon}{1+6\epsilon}}\cdot\frac{3^{k-2}}{2^{k-2}} \le \frac{1-6\epsilon}{6}-\sqrt{\frac{6\epsilon}{1+6\epsilon}}\cdot \frac{2^3}{3^3\beta} \\
    = \ & \frac{1-6\epsilon}{6} - \frac{\sqrt{3(1-6\epsilon)}\cdot 2^2}{3^3} = \frac{9-54\epsilon-8\sqrt{3(1-6\epsilon)}}{54} \\
    < \ & \frac{9-8\sqrt{3(1-6\epsilon)}}{54} < 0
    \end{align*}
    for $\epsilon \le \frac{1}{12}$. Thus, $f(\bar{q}_{k-1})\le \frac{n}{3}$. 
    
    Let now $i=k-1$, which implies that $s > \left(\frac{3}{2}\right)^{k-1}\beta n$. We evaluate $f(\bar{q}_k)$:
    \begin{align*}
           f(\bar{q}_k) \le \ & \frac{3}{4}\left(\frac{1-6\epsilon}{n}\right)n^2 \left[\frac{1}{144}+\epsilon^2+\frac{\epsilon}{6}\right] - \frac{1-6\epsilon}{2}n \left[\frac{1}{12} + \epsilon\right] + \frac{5-6\epsilon}{12}n \\
           & - \frac{2\epsilon}{1+6\epsilon}\left(\frac{3}{2}\right)^{2k-2}n \\
           = \ & \frac{-7(1-6\epsilon)+80-96\epsilon}{192}n-\frac{3\epsilon(1-6\epsilon)}{8}n - \frac{2\epsilon}{1+6\epsilon}\left(\frac{3}{2}\right)^{2k-2}n \\
           = \ & \frac{n}{3}\biggl[\frac{219-162\epsilon - 216\epsilon+1296\epsilon^2}{192}-\frac{6\epsilon}{1+6\epsilon}\left(\frac{3}{2}\right)^{2k-2}\biggr] \\
           = \ & \frac{n}{3}\biggl[1 + \frac{27-378\epsilon+1296\epsilon^2}{192} -\frac{6\epsilon}{1+6\epsilon}\left(\frac{3}{2}\right)^{2k-2}\biggr].
    \end{align*}
    By definition of $k$, we have that
    \begin{align*}
    & \frac{27-378\epsilon+1296\epsilon^2}{192} -\frac{6\epsilon}{1+6\epsilon}\left(\frac{3}{2}\right)^{2k-2} \\ 
    \le \ & \frac{27-378\epsilon+1296\epsilon^2}{192} - \frac{1-6\epsilon}{2}\left(\frac{2}{3}\right)^3 < 0,
    \end{align*}
    where the first and the second inequalities hold for $\epsilon \le \frac{1}{12}$. Thus, $f(\bar{q}_k)\le \frac{n}{3}$. 
    
    We finally evaluate $f(\bar{q})$:
    \begin{align*}
           f\left(\bar{q}\right) \le \ & \frac{n}{12}\left(\frac{(1+3\epsilon)^2}{(1-6\epsilon)}\right) - \frac{n(1+3\epsilon)}{6} + \frac{n(5-6\epsilon)}{12} \\
           = \ & \frac{n}{12}\left(\frac{1-6\epsilon+9\epsilon^2}{1-6\epsilon}-2-6\epsilon+5-6\epsilon\right) \\
           = \ & \frac{n}{12(1-6\epsilon)}\biggl(1-6\epsilon+9\epsilon^2-2+6\epsilon+36\epsilon^2 + 5 -36\epsilon +36\epsilon^2\biggr) \\
           = \ & \frac{n}{12}\left(\frac{4-36\epsilon+81\epsilon^2}{1-6\epsilon}\right) \\
           = \ & \frac{n}{12}\left(\frac{(2-9\epsilon)^2}{1-6\epsilon}\right).
    \end{align*}
   It holds that $f(\bar{q}) \le \frac{n}{3(1-6\epsilon)}$ \big(remember that $\epsilon<\frac{1}{12}$\big); thus, all the evaluations are no more than $\frac{n}{3(1-6\epsilon)}$, and, from an immediate application of the additive Chernoff  bound (Theorem \ref{chernoff:additive} in Appendix \ref{tools}) with $\lambda = \frac{\epsilon n}{1-6\epsilon}$, and by observing that $\bar{q}_i \ge \bar{q}_{i+1}$, we get that
   \[
        \bar{q}_i \le Q \le \frac{n(1+3\epsilon)}{3(1-6\epsilon)},
   \]
   w.h.p.
\end{proof}



At the same time, the following lemma implies that the possible decrease of the bias cannot move it from $S_i$ beyond $S_{i-1}$.

\begin{lemma}\label{lemma2_6}
Let $\mathbf{x}$ be a configuration such that $s\ge \gamma \sqrt{n\log n}$ for some constant $\gamma >0$. If $s> \left(\frac{3}{2}\right)^i \beta n$ for some $0\le i \le k$, then $S>\inf(S_{i-1})$ w.h.p.
\end{lemma}
\begin{proof}[Proof of Lemma \ref{lemma2_6}]
    We have
    \[ \mathbb{E}\left[S\bigm| \mathbf{x}\right] \ge s\left(\frac{5}{6}-\epsilon\right).\]
   The additive form of Chernoff bound (Theorem \ref{chernoff:additive} in Appendix \ref{tools}) on $A$ and $B$, together with the union bound, implies that $S\ge s\left(\frac{5}{6}-2\epsilon\right)$ w.h.p. (we can proceed as in the proof of Lemma \ref{lemma:thm1:biasincrease}). Thus
    \[
        s\left(\frac{5}{6}-2\epsilon\right) >  \left(\frac{3}{2}\right)^i \beta n \left(\frac{5}{6}-2\epsilon\right) \ge \left(\frac{3}{2}\right)^{i-1} \beta n = S_{i-1}
    \]
    since $\frac{3}{2}\left(\frac{5}{6}-2\epsilon\right)\ge 1$.
\end{proof}



We still need to ``control'' the dynamics, in particular the case in which $q> \frac{n(1+3\epsilon)}{3(1-6\epsilon)}$. The next lemma fulfill this need. It shows that there is a decrease of the number of undecided nodes when they are more than $\frac{n(1+3\epsilon)}{3(1-6\epsilon)}$, and provides a lower bound on the decrease, depending on the bias.

\begin{lemma}\label{lemma2_7}
    Let $\mathbf{x}$ be a configuration such that $q\ge \frac{n}{3(1-6\epsilon)}$. Then, it holds that
    \begin{itemize}
        \item[(1)] $Q\le q\left(1-\epsilon\right)$ w.h.p.;
        \item[(2)] if $s\le \sup(S_i)$ for some $-1\le i \le k-1$, then $Q\ge \bar{q}_{i+1}$ w.h.p.
    \end{itemize}
\end{lemma}
\begin{proof}[Proof of Lemma \ref{lemma2_7}]
    Consider the first item. We define $f(q)=\frac{3}{4}\left(\frac{1-6\epsilon}{n}\right)q^2 - \frac{1-6\epsilon}{2}q + \frac{5-6\epsilon}{12}n$. We now show that $f(q)\le q\left(1-2\epsilon\right)$. Indeed, $f(q) - q\left(1-2\epsilon\right)$ is equal to
    \begin{align*}
            & \frac{3}{4}\left(\frac{1-6\epsilon}{n}\right)q^2 - \left(\frac{1-6\epsilon}{2}+1-2\epsilon\right)q + \frac{5-6\epsilon}{12}n \\
            = \ & \frac{3}{4}\left(\frac{1-6\epsilon}{n}\right)q^2 - \left(\frac{3-10\epsilon}{2}\right)q + \frac{5-6\epsilon}{12}n.
    \end{align*}
    This expression is a convex parabola which has its maximum in either $\bar{q}_1=\frac{n}{3(1-6\epsilon)}$ or $\bar{q}_2=n$. We calculate $f(q)-q\left(1-2\epsilon\right)$ in these two points
    \begin{align*}
            & \frac{3}{4}\left(\frac{1-6\epsilon}{n}\right)q_1^2 - \left(\frac{3-10\epsilon}{2}\right)q_1+ \frac{5-6\epsilon}{12}n \\ 
            = \ & \frac{n^2}{12(1-6\epsilon)}\left(1-6+20\epsilon+5-36\epsilon+36\epsilon^2\right) \\
            = \ & \frac{n^2}{12(1-6\epsilon)}\left[-4\epsilon(4-9\epsilon)\right] <0
    \end{align*}
    for all $0<\epsilon\le \frac{1}{12}$. At the same time it holds that
    \begin{align*}
            & \frac{3}{4}\left(\frac{1-6\epsilon}{n}\right)q_2^2 - \left(\frac{3-10\epsilon}{2}\right)q_2+ \frac{5-6\epsilon}{12}n\\
            = \ & \frac{n}{12}\left[9-54\epsilon-18+60\epsilon+5-6\epsilon\right] \\
            = \ & -\frac{n}{3}<0.
    \end{align*}
    Thus, $\mathbb{E}\left[Q\mid \mathbf{x}\right] \le q(1-2\epsilon)$. The additive form of Chernoff bound (Theorem \ref{chernoff:additive} in Appendix \ref{tools}) with $\lambda = \frac{\epsilon n}{3(1-6\epsilon)}$ implies that $Q < q(1-\epsilon)$ w.h.p.
    
    As for the second item, We consider two cases. First, assume $i<k-1$, thus $s\le \left(\frac{3}{2}\right)^{i+1}\beta n$.
    Take the expectation \eqref{expectation_Q_toomuchnoise} and observe that
    \begin{align*}
            \mathbb{E}\left[Q\bigm|\mathbf{x}\right] \ge \ & \frac{3}{4}\left(\frac{1-6\epsilon}{n}\right)\frac{n^2}{9(1-6\epsilon)^2} +\frac{5-6\epsilon}{12}n -\left(\frac{1-6\epsilon}{2}\right)\frac{n}{3(1-6\epsilon)} - \frac{1-6\epsilon}{n}\left(\frac{s}{2}\right)^2  \\
            \ge\ & \frac{n}{12}\left(\frac{4-24\epsilon+36\epsilon^2}{1-6\epsilon}\right) - \frac{2n\epsilon}{1+6\epsilon}\left(\frac{3}{2}\right)^{2i+2}  \\
            \ge \ & \frac{n}{3(1-6\epsilon)}\left(1-3\epsilon\right)^2 - \frac{2n\epsilon}{1+6\epsilon}\left(\frac{3}{2}\right)^{2i+2}  \\
            \ge \ & \frac{n}{3}-\frac{2n\epsilon}{1+6\epsilon}\left(\frac{3}{2}\right)^{2i+5}+\frac{19n\epsilon}{4}\left(\frac{3}{2}\right)^{2i+2} \\
            = \ & \bar{q}_{i+1}+\frac{19n\epsilon}{4}\left(\frac{3}{2}\right)^{2i+2} .
    \end{align*}
    We conclude applying the additive Chernoff bound (Theorem \ref{chernoff:additive} in Appendix \ref{tools}) with $\lambda = \frac{19n\epsilon}{4}\left(\frac{3}{2}\right)^{2i+2}$, obtaining $Q \ge \bar{q}_{i+1}$, w.h.p.
    
    Second, let $i=k-1$; then $s\le \frac{2}{3}n$. As before
    \begin{align*}
            \mathbb{E}\left[Q\bigm|\mathbf{x}\right] \ge \ & \frac{3}{4}\left(\frac{1-6\epsilon}{n}\right)\frac{n^2}{9(1-6\epsilon)^2}+\frac{5-6\epsilon}{12}n -\left(\frac{1-6\epsilon}{2}\right)\frac{n}{3(1-6\epsilon)} - \frac{1-6\epsilon}{n}\left(\frac{s}{2}\right)^2 \\
            \ge \ & \frac{n}{12(1-6\epsilon)}\left(4-24\epsilon+36\epsilon^2\right) - \frac{n(1-6\epsilon)}{9}  \\
            \ge \ & \frac{n\left(1-3\epsilon\right)^2}{3(1-6\epsilon)} - \frac{n(1-6\epsilon)}{9} \\
            = \ & \frac{n}{9(1-6\epsilon)}\left(2-6\epsilon-9\epsilon^2\right) \\
            = \ &\frac{n}{12}\left(\frac{4(2-6\epsilon-9\epsilon^2)}{3(1-6\epsilon)}\right) \\
            \ge \ & \frac{n}{12}(1+14\epsilon) = \bar{q}_k+\frac{\epsilon n}{6}
    \end{align*}
    since \[\frac{4(2-6\epsilon-9\epsilon^2)}{3(1-6\epsilon)} -(1+14\epsilon) >\frac{5-48\epsilon}{3(1-6\epsilon)}>0\] for $\epsilon\le \frac{1}{12}$.
    Thus, we conclude applying the additive form of Chernoff  bound (Theorem \ref{chernoff:additive} in Appendix \ref{tools}) with $\lambda = \frac{\epsilon}{6}n$, obtaining $Q \ge \bar{q}_k$, w.h.p.
\end{proof}

Next and last lemma guarantees that once the process reaches a configuration having bias $\bigo(\sqrt{n\log n})$, then it ``enters'' a metastable phase that lasts $\Omega(n^{\lambda'})$ rounds w.h.p. in which the absolute value of the bias remains $\bigo(\sqrt{n\log n})$, since it can be used symmetrically when the bias is negative.

\begin{lemma}\label{lemma2_9}
    Let $\mathbf{x}$ be any configuration. If $s \le \gamma\sqrt{n\log n}$ for some constant $\gamma >0$ and $q\le \frac{n}{3}\left(\frac{1+3\epsilon}{1-6\epsilon}\right)$, then $S\le 2\gamma\sqrt{n\log n}$ w.h.p.
\end{lemma}
\begin{proof}[Proof of Lemma \ref{lemma2_9}]
    Consider the expectation of $S$ from equation \ref{expectation_S_toomuchnoise}. We have
    \begin{align*}
        \mathbb{E}\left[S\bigm| \mathbf{x}\right] \le\ & \gamma\sqrt{n\log n} \left(\frac{5}{6}-\epsilon+\frac{1}{6}+\frac{\epsilon}{2}\right) \\
        \le\ & \gamma\sqrt{n\log n}.
    \end{align*}
    We conclude applying the additive form of Chernoff bound (Theorem \ref{chernoff:additive} in Appendix \ref{tools}), and the union bound, as we did in the proof of Theorem \ref{lemma:thm2:bigepsilon}.
\end{proof}

\begin{proof}[Proof of Theorem \ref{theorem_victory_of_noise}]
    If $\frac{1}{12}< \epsilon$ the theorem is true due to Lemma \ref{lemma:thm2:bigepsilon}. Let us assume $0<\epsilon\le\frac{1}{12}$. The proof is divided into different cases (recall Lemma \ref{lemma:whp_intersection} in the preliminaries). Let $\gamma>0$ be any constant.
    
    \begin{itemize}
        \item[(1)] $s\in S_i$ for some $-1\le i\le k-1$ and $s\ge \gamma\sqrt{n\log n}$,
        \begin{itemize}
            \item[(1.1)] $\bar{q}_{i+1}\le q \le \frac{n(1+3\epsilon)}{3(1-6\epsilon)}$: the bias decreases exponentially fast each round, w.h.p., until $0\le s\le \gamma\sqrt{n\log n}$ due to the combination of Lemmas \ref{lemma2_3}, \ref{lemma2_4}, \ref{lemma2_5}, and \ref{lemma2_6}. This phase lasts $O(\log n)$ rounds;
            \item[(1.2)] $\frac{n(1+3\epsilon)}{3(1-6\epsilon)} < q$: Lemmas \ref{lemma2_7} and \ref{lemma2_4} imply that in $\bigo(\log n)$ rounds the number of undecided nodes reaches the interval $\left[\bar{q}_{i+1},\frac{n}{3(1-6\epsilon)}\right]$ where $i$ is such that the round before the undecided nodes become less than $\frac{n}{3(1-6\epsilon)}$, the bias is in $S_i$: remind that during the whole process (which lasts $\bigo(\log n)$ rounds) the bias never goes over $\frac{2}{3}n$ thanks to Lemma \ref{lemma:thm2:biascontrol}. 
            At the same time, the bias will be in one set between $S_{i-1}, \dots, S_{k-1}$ due to Lemma \ref{lemma2_6}. Since $\bar{q}_i$ is a decreasing sequence, we are in Case 1.1, and we conclude.
            \item[(1.3)] $q < \bar{q}_{i+1}$: in this case, Lemma \ref{lemma2_4} implies $Q\ge \bar{q}_{i+1}$ in the next round, w.h.p. Since Lemma \ref{lemma:thm2:biascontrol} guarantees that the bias remains under the value $\frac{2}{3}n$ w.h.p., either we are in Case 1.1 or in Case 1.2, and we conclude.
        \end{itemize}
        \item[(2)] $s > \frac{2}{3}n$: Lemma \ref{lemma:thm2:biascontrol} implies that the bias gets less than or equal to $\frac{2}{3}n$ in $\bigo(\log n)$ rounds, w.h.p.; then we are in Case 1 and we conclude.
    \end{itemize}
    
    Now, we can suppose the process starts from a configuration $\mathbf{y}$ having bias $0\le s(\mathbf{y})\le \gamma\sqrt{n\log n}$, and such that $\bar{q}_0\le q(\mathbf{y}) \le \frac{n}{3}\left(\frac{1+3\epsilon}{1-6\epsilon}\right)$, as Case 1.1 or 1.3 leaves it. In the next round, it holds that the number of undecided nodes is $ \bar{q}_{-1} \le Q\le \frac{n}{3}\left(\frac{1+3\epsilon}{1-6\epsilon}\right)$, w.h.p., due to Lemma \ref{lemma2_5}; at the same time, w.h.p., $\lvert S \rvert \le 2\gamma\sqrt{n\log n}$ for Lemma \ref{lemma2_9} (which can be used symmetrically on $A-B$ and $B-A$). Thus, the absolute value of the bias is either still less than $\gamma\sqrt{n\log n}$ or has become greater than or equal to $\lvert S \rvert \ge \gamma\sqrt{n\log n}$, in which case it starts decreasing exponentially fast each round, w.h.p., for Lemma \ref{lemma2_3} until becoming again less than $\gamma\sqrt{n\log n}$ (as explained in Case 1.1, which works analogously if the bias is negative, because of symmetry). This phase lasts $\Omega(n^{\lambda'})$ rounds, for some sufficiently small constant $\lambda' >0$ (see Lemma \ref{lemma:whp_intersection} in the preliminaries). As in the proof of Theorem \ref{theorem_almost_plurality}, $\lambda'$ depends only on $\gamma$ and $\epsilon$. 
\end{proof}


\section{ Symmetry Breaking from Balanced Configurations} \label{sec:symbreak}

IIn this section,  we consider the  \uproc   starting from arbitrary initial configurations: in particular,   from   configurations having no bias.
Interestingly enough, we  show   
a  transition phase similar to that proved in the previous section. Informally, the next theorem states that when $p < 1/6$, the \uproc is able to break the symmetry of any perfectly-balanced   initial configuration and to compute  almost consensus within $O(\log n)$ rounds, w.h.p.  

\begin{restatable}{theorem}{thmsymbre}\label{thm:symbre}
Let $\mathbf{x}$ be any initial configuration, and let $\epsilon\in\left(0,1/6\right)$ be some absolute  constant. If $p=1/6 - \epsilon$ is the noise probability, then the \uproc reaches a configuration $\mathbf{y}$ having   bias $s$ toward some opinion $j \in \{\mesalpha,\mesbeta\}$ such that   $|s(\mathbf{y})| \in \Delta = \left[\frac{2\sqrt{\epsilon}}{1+6\epsilon}n, \left(1-2\left(\frac{1-6\epsilon}{12}\right)^3\right)n\right]$  within $\bigo(\log n)$ rounds, w.h.p. 
Moreover, once reached configuration $\mathbf{y}$, the \uproc enters  a (metastable) phase of length $\Omega(n^\lambda)$ rounds (for some  constant $\lambda>0$) where  the majority opinion is $j$ and the bias  
keeps within  the range $\Delta$, w.h.p.
\end{restatable}

What follows is an outline of the proof of the theorem, while more details are given in the next subsection.
\begin{proof}[Outline of Proof of Theorem \ref{thm:symbre}]
  If the initial configuration   $\mathbf{x}$ has bias $s = \Omega(\sqrt{n\log n})$ then the claim of the theorem is equivalent to that of Theorem \ref{theorem_almost_plurality}, so we are done.
  Hence, we next   assume the initial bias $s$ be   $o(\sqrt{n\log n})$: for  this case, our proof proceeds along the following main   steps.
   
   \noindent
   \emph{Step I.} Whenever the   bias $s$ is   small, i.e. 
   $o(n)$,  we prove that, within the next $\bigo(\log n)$ rounds,  the number of undecided nodes turns out to keep always in a suitable linear range: roughly speaking, we get that this number lies in
   $(n/3, n/2]$, w.h.p. 
   
   \noindent
   \emph{Step II.}    Whenever $s$ is very small, i.e. $s= o(\sqrt n)$, 
   there is no effective  drift towards any opinion. However,   we can prove   that, thanks to Step I,  the random variable $S$, representing the bias in  the next round, has   \emph{high variance}, i.e. $\Theta(n)$.  The latter holds since  $S$ can be written as a suitable  sum whose addends  include some    random variables having binomial distribution of expectation $0$:  so,  we can  apply the  Berry-Essen Theorem (Theorem \ref{thm:berryeseen} in Appendix \ref{tools}) to get a  lower bound on the variance of $S$. 
   Then, thanks to this large variance, standard arguments for the standard deviation imply that, in this parameter range, there is a positive constant probability that $S$ will get some value of magnitude $\Omega(\sqrt{n})$ (see Claim 1 of Lemma \ref{lemma:symbreak}). Not surprisingly,  in this phase,  we find out that the variance of $S$ is not decreased by the communication noise. We can thus claim  that   the  process,   at every round, 
     has positive constant probability to reach a configuration having bias $s = \omega(\sqrt{n})$ and $q \in (n/3, n/2]$. Then, after $O(\log n)$ rounds, this event will happen w.h.p. 
   
   \noindent
   \emph{Step III.} Once the process reaches a configuration with $s = \omega(\sqrt{n})$ and $q \in (n/3, n/2]$,   we then prove that 
   the expected bias increases  by a constant factor (which depends on $\epsilon$).   Observe that we cannot   use here the same round-by-round concentration argument that works for bias over $\sqrt{n\log n}$ (this is in fact the  minimal magnitude required to apply the Chernoff's bounds \cite{chernoff1}). We instead exploit a useful general tool \cite{DBLP:conf/mfcs/ClementiGGNPS18} that bounds   the stopping time of  some class of Markov chains having rather mild conditions on the drift towards their absorbing states (see Lemma \ref{lemma:symmetrygeneric}). 
   This tool in fact allows us to consider  the two phases described, respectively,  in Step II and Step III  as a unique \emph{symmetry-breaking} phase of the process. Our  final 
   technical contribution here is to show that the conditions required to apply this tool   hold whenever the communication noise parameter is such that $p \in (0,1/6)$. 
    This allows us to prove that, within $O(\log n)$ rounds, the process reaches a configuration with bias $s = \Omega(\sqrt{n}\log n)$, w.h.p.
\end{proof}

\noindent
\textbf{Large communication noise (the case $p > 1/6+\epsilon$).}
When $p > 1/6+\epsilon$,    Theorem \ref{theorem_victory_of_noise}   a fortiori holds when the initial  bias is small, i.e. $s = o(\sqrt{n \log n})$: thus, we get that, in this case,     the system enters into a long regime of non consensus, starting from any initial configuration. Then, by combining the results for biased configurations in Section \ref{sec:phasetrans} with those in this section, we can 
observe the  phase transition of the \uproc  starting from any possible initial configuration.

\begin{restatable}{theorem}{thmunbiasedvictorynoise}\label{thm:unbiased_victory_noise}
Let $\config$ be any initial configuration, and let  $\epsilon\in\left(0,1/3\right]$ be some absolute constant. If
$p=1/6 + \epsilon$ is the  noise probability, then the \uproc reaches a configuration $\mathbf{y}$ having bias $\lvert s(\mathbf{y})\rvert = \bigo(\sqrt{n\log n})$ within $\bigo(\log n)$ rounds, w.h.p. Furthermore, starting from such a configuration, the \uproc enters a (metastable) phase of length $\Omega\left(n^{\lambda'}\right)$ rounds (for some constant $\lambda'>0$) where the absolute value of the bias  keeps bounded by $\bigo(\sqrt{n\log n})$, w.h.p.
\end{restatable}

\smallskip
\noindent
\textbf{Stubborn agents.}
We conclude this section by observing that the equivalence result shown in  Lemma \ref{lem:equiv}  holds independently of the  choices of  the noise parameter $p \in (0, 1/2]$, and  of the initial bias: the phase transition
of the \uproc in the presence of stubborn agents thus 
holds even in the case of unbiased configurations.
\begin{corollary}\label{cor:unbiased_stubborn}
    Let $\frac 12 >p>0$ be a constant, and let the \stubproc start from any initial configuration.
        If $p < \frac 16$, then, in $\bigo(\log n)$ rounds, the \stubproc enters   a metastable phase of almost consensus towards some opinion $j\in\{\mesalpha, \mesbeta\}$ of length $\Omega\left(n^\lambda\right)$ for some constant $\lambda > 0$,  in which the absolute value of the bias is $\Theta(n)$, w.h.p.
        If $p \in (\frac 16, \frac 12]$, then, in $\bigo(\log n)$ rounds, the \stubproc enters   a metastable phase of length $\Omega\left(n^{\lambda'}\right)$ for some constant $\lambda' > 0$ where the absolute value of the bias keeps bounded by $\bigo(\sqrt{n\log n})$, w.h.p.
\end{corollary}

\subsection{Proof of Theorem \ref{thm:symbre}: More details} \label{ssec:symbreakproof}

The proof of Theorem \ref{thm:symbre} essentially relies on the following lemma which has been proved in  \cite{DBLP:conf/mfcs/ClementiGGNPS18} (we report a proof here).

\begin{lemma}\label{lemma:symmetrygeneric}
    Let $\{X_{t}\}_{t\in \mathbb{N}}$ be a Markov Chain with finite-state space $\Omega$ and let $f:\Omega\mapsto[0,n]$ be a function that maps states to integer values. Let $c_3$ be any positive constant and let $m = c_3\sqrt{n}\log n$ be a target value. Assume the following properties hold:
    \begin{itemize}
        \item[(1)] for any positive constant $h$, a positive constant $c_1 < 1$ exists such that for any $x \in \Omega : f(x) < m$,
        \[\mathbb{P}\{f(X_{t+1}) < h\sqrt{n} \mid X_{t} = x\} < c_1;\]
        \item[(2)] there exist two positive constants $\delta$ and $c_2$ such that for any $x \in \Omega: h\sqrt{n} \leq f(x) < m$,
        \[\mathbb{P}\{f(X_{t+1}) < (1+\delta)f(X_{t}) \mid X_{t} = x\} < e^{-c_2f(x)^2/n}.\]
    \end{itemize}
    Then the process reaches a state $x$ such that $f(x) \ge m$ within $\bigo(\log n)$ rounds, w.h.p.
\end{lemma}
\begin{proof}[Proof of Lemma \ref{lemma:symmetrygeneric}]
    Define a set of hitting times $T \coloneqq \left\{\tau(i)\right\}_{i \in \nat}$, where
    \[\tau(i) = \inf_{i \in \nat}\left\{t: t > \tau(i-1), f(X_t) \ge h\sqrt{n}\right\},\]
    setting $\tau(0) = 0$. By the first hypothesis, for every $i\in \nat$, the expectation of $\tau(i)$ is finite. Now, define the following stochastic process which is a subsequence of $\{X_t\}_{t\in \nat}$:
    \[\{R_i\}_{i \in \nat} = \{X_{\tau(i)}\}_{i\in \nat}.\] Observe that $\{R_i\}_{i \in \nat}$ is still a Markov chain. Indeed, if $\{x_1, \dots, X_{i-1}\}$ be a set of states in $\Omega$, then
    \begin{align*}
        & \ \prob(R_i = x \mid R_{i-1} = x_{i-1}, \dots, R_1 = x_1) \\ 
        = & \ \prob(X_{\tau(i)} = x \mid X_{\tau(i-1)} = x_{i-1}, \dots, X_{\tau(1)} = x_1) \\
        = & \ \sum_{t(i)> \dots> t(1) \in \nat} \prob(X_{t(i)} = x \mid X_{t(i-1)}=x_{i-1}, \dots, X_{t(1)}=x_1)\\
        & \cdot \prob\left(\tau(i)=t(i), \dots, \tau(1)=t(1)\right) \\
        = & \ \prob(X_{\tau(i)}=x \mid X_{\tau(i-1)=x_{i-1}}) \\
        = & \ \prob(R_i = x \mid R_{i-1} = x_{i-1}).
    \end{align*}
    By definition, the state space of $R$ is $\{x \in \Omega : f(x) \ge h\sqrt{n}\}$. Moreover, the second hypothesis still holds for this new Markov chain. Indeed: 
    \begin{align*}
        & \ \prob\left(f(R_{i+1}< (1+\epsilon)f(R_i) \mid R_i = x \right) \\
        = & \ 1- \prob\left(f(R_{i+1} \ge (1+\epsilon)f(R_i) \mid R_i = x\right) \\
        = & \ 1- \prob\left(f(X_{\tau(i+1)} \ge (1+\epsilon)f(X_{\tau(i)}) \mid X_{\tau(i)} = x\right) \\
        \le & \ 1- \prob\left(f(X_{\tau(i+1)} \ge (1+\epsilon)f(X_{\tau(i)}), \tau(i+1) = \tau(i) +1 \mid X_{\tau(i)} = x\right) \\
        = & \ 1- \prob\left(f(X_{\tau(i)+1} \ge (1+\epsilon)f(X_{\tau(i)}) \mid X_{\tau(i)} = x\right) \\
        = & \  1- \prob\left(f(X_{t+1} \ge (1+\epsilon)f(X_{t})\mid X_{t} = x\right) \\
        < & \ e^{-c_2f(x)^2/n}.
    \end{align*}
    These two properties are sufficient to study the number of rounds required by the new Markov chain $\{R_i\}_{i\in \nat}$ to reach the target value $m$. Indeed, by defining the random variable $Z_i = \frac{f(R_i)}{\sqrt{n}}$, and considering the following ``potential'' function, $Y_i = exp\left(\frac{m}{\sqrt{n}}-Z_i\right)$, we can compute its expectation at the next round as follows. Let us fix any state $x\in \Omega$ such that $h\sqrt{n} \le f(x) < m$, and define $z = \frac{f(x)}{\sqrt{n}}$, $y=exp\left(\frac{m}{\sqrt{n}}-z\right)$. We have
    \begin{align}
    \mean[Y_{i+1} \vert R_i = x] & \leq \prob\left(f(R_{i+1}) < (1+\epsilon)f(x)\right) e^{m/\sqrt{n}}\nonumber\\
    &+ \prob\left(f(R_{i+1}) \geq (1+\epsilon)f(x)\right) e^{m/\sqrt{n} - (1 + \epsilon)z}\nonumber\\
    \mbox{ (from Hypothesis  (2)) } &\leq e^{-c_2 z^2} \cdot  e^{m/\sqrt{n}} + 1 \cdot e^{m/\sqrt{n} - (1 + \epsilon)z}\nonumber\\
    &= e^{m/\sqrt{n} - c_2 z^2} + e^{m/\sqrt{n} - z - \epsilon z}\nonumber\\
    &= e^{m/\sqrt{n} - z} (e^{z - c_2 z^2} + e^{-\epsilon z})\nonumber\\
    &\leq e^{m/\sqrt{n} - z}(e^{-2} + e^{-2})\label{eq:smallh}\\
    &<\frac{e^{m/\sqrt{n} - z}}{e}\nonumber\\
    &=\frac{y}{e}\nonumber,
    \end{align}
    where in~\eqref{eq:smallh} we used that $z$ is always at least $h$ and thanks to Hypothesis (1) we can choose a sufficiently large $h$.
    
    By applying the Markov inequality and iterating the above bound, we get
    \[\prob(Y_{i} > 1) \leq \frac{\mean[Y_{i}]}{1} \leq \frac{\mean[Y_{{i}-1}]}{e}\leq \cdots \leq \frac{\mean[Y_0]}{e^{\tau_R}} \leq \frac{e^{m/\sqrt{n}}}{e^{i}}.\]
    We observe that if $Y_{i} \leq 1$ then $R_{i} \geq m$, thus by setting ${i} = m/\sqrt{n} + \log n = (c_3 + 1)\log n$, we get:

    \begin{equation}
    \prob\left(R_{(c_3 + 1)\log n} < m\right) = \prob\left(Y_{(c_3 + 1)\log n} > 1\right) < \frac{1}{n}.\label{eq:markovR}
    \end{equation}
    
    Our next goal is to give an upper bound on the hitting time $\tau_{(c_3 + 1)\log n}$. 
    Note that the event ``$\tau_{(c_3 + 1)\log n} > c_4\log n$'' holds if and only if the number of rounds such that $f(X_t) \geq h\sqrt{n}$ (before round $c_4\log n$) is less than $(c_3 + 1)\log n$.
    Thanks to Hypothesis (1), at each round $t$ there is at least probability $1-c_1$ that $f(X_{t}) \geq h\sqrt{n}$. This implies that, for any positive constant $c_4$, the probability $\prob\left(\tau_{(c_3 + 1)\log n} > c_4\log n\right)$ is bounded by the probability that, within $c_4\log n$ independent Bernoulli trials, we get less then $(c_3 + 1)\log n$ successes, where the success probability is at least $1-c_1$. We can thus choose a sufficiently large $c_4$ and apply the multiplicative form of the Chernoff bound (Theorem \ref{chernoff:multiplicative} in Appendix \ref{tools}), obtaining
    
    \begin{equation}
    \prob\left(\tau_{(c_3 + 1)\log n} > c_4\log n\right) < \frac{1}{n}.\label{eq:tauR}
    \end{equation}
    
    We are now ready to prove the Lemma using Inequalities \eqref{eq:markovR} and \eqref{eq:tauR}, indeed
    
    \begin{align*}
    \prob\left(X_{c_4 \log n} \geq m\right)
    &> \prob\left(R_{(c_3 + 1)\log n} \geq m \wedge \tau_{(c_3 + 1)\log n} \leq c_4 \log n\right)\\
    &= 1 - \prob\left(R_{(c_3 + 1)\log n} < m  \vee \tau_{(c_3 + 1)\log n} > c_4 \log n\right)\\
    &\geq 1 - \prob\left(R_{(c_3 + 1)\log n} < m\right) + \prob\left(\tau_{(c_3 + 1)\log n} > c_4 \log n\right)\\
    &> 1 - \frac{2}{n}.
    \end{align*}
    
    Hence, choosing a suitable big $c_4$, we have shown that in $c_4 \log n$ rounds the process reaches the target value $m$, w.h.p.
\end{proof}

Our goal is 
to apply the above lemma to the \uproc (which defines a finite-state Markov chain) starting with bias of size $o(\sqrt{n\log n})$ where we set $f(\mathbf{X}_t)=s(\mathbf{X}_t)$, $c_3 = \gamma>0$ for some constant $\gamma > 0$, and $m=\gamma\sqrt{n}\log n$: this would 
imply the upper bound $\bigo(\log n)$ on the number of rounds needed to reach a configuration having bias $\Omega(\sqrt{n\log n})$, w.h.p., breaking the symmetry because Theorem \ref{theorem_almost_plurality} then holds. To this aim, with the next two lemmas we show that the  \uproc satisfies the hypotheses of Lemma \ref{lemma:symmetrygeneric} in this setting, w.h.p.

\begin{lemma}\label{lemma_3_1}
Let $\mathbf{x}$ be any configuration in which $s \le \beta n$ and $\frac{n}{3}\left(\frac{1-4\epsilon}{1+6\epsilon}\right) \le q \le \frac{n}{2}$. Then, in the next round, it holds that $\frac{n}{3}\left(\frac{1-4\epsilon}{1+6\epsilon}\right) \le Q \le \frac{n}{2}$ w.h.p.
\end{lemma}
\begin{proof}[Proof of Lemma \ref{lemma_3_1}]
    The fact that $Q \ge \frac{n}{3}\left(\frac{1-4\epsilon}{1+6\epsilon}\right)$ with probability $1-\exp(\Theta(n))$ (which is w.h.p.) comes from Lemma \ref{lemma:thm1:enoughundecided}.
    At the same time, from Equation \ref{expectation_Q}, it holds that
    \[
        \mean[Q \mid \mathbf{x}] \le \frac{(1-6\epsilon)n}{6} + \frac{1+6\epsilon}{4n}\left(2q^2+(n-q)^2\right).
    \] 
    Denote this expression as $f(q)$. For $\frac{n}{3}\left(\frac{1-4\epsilon}{1+6\epsilon}\right) \le q \le \frac{n}{2}$, the maximum of $f$ is obtained either at $q_1=\frac{n}{3}\left(\frac{1-4\epsilon}{1+6\epsilon}\right)$ or at $q_2=\frac{n}{2}$. 
    
    \begin{align*}
        f(q_1) = & \ \frac{1-6\epsilon}{3}n+\frac{2(1-4\epsilon)^2+4(1+11\epsilon)^2}{3^2\cdot 4(1+6\epsilon)}n \\
        = & \ \frac{6(1-36\epsilon^2)+3+36\epsilon+258\epsilon^2}{2\cdot 3^2(1+6\epsilon)}n \\
        = & \ \frac{9 +36\epsilon+42\epsilon^2}{2\cdot 9 (1+6\epsilon)}n \\
        = & \ \frac{1 + 4\epsilon + \frac{14}{3}\epsilon^2}{2(1+6\epsilon)}n \\
        \le & \ \left(\frac{1}{2}-c\right)n, \\
        f(q_2) = & \ \frac{1-6\epsilon}{6}n+\frac{3(1+6\epsilon)}{16}n \\
        = & \ \frac{17+6\epsilon}{48}n \\
        \le & \ \left(\frac{1}{2}-c\right)n,
    \end{align*}
    for some suitable constant $c>0$, and for $\epsilon < \frac{1}{6}$. Thus, $f(q)\le (1/2 - c)n$ and we conclude by using the additive form of Chernoff bound (Theorem \ref{chernoff:additive} in Appendix \ref{tools}), getting that $Q \le \frac{n}{2}$, with probability $1-\exp\left(c^2 n\right)$. Then, the intersection of two events holding with probability $1-\exp(\Theta(n))$ is still an event holding with probability $1-\exp(\Theta(n))$.
\end{proof}

\begin{lemma}\label{lemma:symbreak}
    Let $\mathbf{x}$ be any configuration such that $q(\mathbf{x})\in \left[\frac{n}{3}\left(\frac{1-4\epsilon}{1+6\epsilon}\right), \frac{n}{2}\right]$. Then, it holds that
    \begin{itemize}
        \item[(1)] for any constant $h>0$ there exists a constant $c_1>0$ such that 
        \[\mathbb{P}(\lvert S \rvert < h\sqrt{n}) \mid \mathbf{X}_t = \mathbf{x}) < c_1;\]
        \item[(2)] there exist two positive constants $\delta$ and $c_2$ such that 
        \[\mathbb{P}(\lvert S \rvert \ge (1+\delta)s \mid \mathbf{X}_t = \mathbf{x}) \ge 1-e^{-c_2\frac{s^2}{n}}.\]
    \end{itemize}
\end{lemma}
\begin{proof}[Proof of Lemma \ref{lemma:symbreak}]
    As for the first item, let $\mathbf{x}$ and $\mathbf{x}_0$ be two states such that $\lvert s(\mathbf{x})\rvert < h \sqrt{n}$, $\lvert s(\mathbf{x}_0)\rvert=0$, $q(\mathbf{x})=q(\mathbf{x}_0)$. A simple domination argument implies that 
    \[\mathbb{P}(\lvert S \rvert < h\sqrt{n}\mid \mathbf{X}_t = \mathbf{x}) \le \mathbb{P}(\lvert S \rvert < h\sqrt{n}\mid \mathbf{X}_t = \mathbf{x}_0).\] 
    Thus, we can bound just the second probability, where the initial bias is zero, which implies that $a=b$.
    
    Define $A^q$, $B^q$, $Q^q$ the random variables counting the nodes that were undecided in the configuration $\mathbf{x}_0$ and that, in the next round, get the opinion \mesalpha, \mesbeta, and undecided, respectively. Similarly, $A^a$ ($B^b$) counts the nodes that support opinion \mesalpha (\mesbeta) in the configuration $\mathbf{x}_0$ and that, in the next round, still support the same opinion. Trivially, $A=A^q+A^a$ and $B=B^q+B^b$. Moreover, observe that, among these random variables, only $A^q$ and $B^q$ are mutually dependent. Thus, conditioned to the event $\{\mathbf{X}_t=\mathbf{x}_0)\}$, if $\alpha = \mathbb{E}[A^a\mid \mathbf{X}_t = \mathbf{x}_0] = \mathbb{E}[B^b\mid \mathbf{X}_t = \mathbf{x}_0]$, it holds that
    \begin{align*}
        \mathbb{P}(\lvert S \rvert \ge h\sqrt{n})
        \ge & \ \mathbb{P}(A\ge B + h\sqrt{n}) \\
        \ge & \ \mathbb{P}(A^q\ge B^q + h\sqrt{n})\mathbb{P}(A^a\ge \alpha)\mathbb{P}(B^b\le \alpha).
    \end{align*}
    The random variables $A^a-\alpha$ and $B^b-\alpha$ happen to be binomial distribution with expectation 0 (recall that $a=b$), and finite second and third moment. Thus, the Berry-Essen Theorem (Theorem \ref{thm:berryeseen} in Appendix \ref{tools}) allows us to approximate up to an arbitrary-small constant $\epsilon_1>0$ (as long as $n$ is large enough) both the random variables with a normal distribution that has expectation 0. Thus, 
    \[\mathbb{P}(A^a\ge \alpha)=\mathbb{P}(B^b\le \alpha)\ge \left(\frac{1}{2}-\epsilon_1\right).\]
    As for the random variable $A^q-B^q$, notice that conditioned to the event $\{q-Q^q = k\}$, it is the sum of $k$ Rademacher random variables. The hypothesis $q\le \frac{n}{2}$ allows us to use the Chernoff bound on $Q^q$ and show that $Q^q\le \frac{3}{4}q$ w.h.p. Thus, since $q\ge \frac{n}{3}\left(\frac{1-4\epsilon}{1+6\epsilon}\right)$, it holds that $q-Q^q = \Theta(n)$ w.h.p. It follows that the conditional variance of $A^q-B^q$ given $q-Q^q$ yields $\Theta(n)$ w.h.p., and $A^q-B^q$ conditioned to the event $E = \{q-Q^q = \Theta(n)\}$ can be approximated by a normal distribution up to an arbitrary-small constant $\epsilon_2>0$. Then, we have that 
    \[\mathbb{P}(A^q\ge B^q + h\sqrt{n}) \ge \mathbb{P}(A^q\ge B^q + h\sqrt{n}\mid E)\mathbb{P}(E)\ge \epsilon_2.\]
    Setting $c_1 = \epsilon_1 \cdot \epsilon_2$, we get property (1). 
    
    As for property (2), it is easy to see that the hypothesis on $q$ implies that 
    \[\mathbb{E}[S \mid \mathbf{X}_t = \mathbf{x}] \ge s\left(1+\frac{\epsilon}{3}\right).\] We can get the property applying the additive Chernoff bound (Theorem \ref{chernoff:additive} in Appendix \ref{tools}) separately on $A$ and $B$, and then the union bound, as we did in the proof of Lemma \ref{lemma:thm1:biasincrease}, getting that $\expect{S \mid \config} \ge s(1+\epsilon/6)$, w.h.p.
\end{proof}

The reader may notice that Lemma \ref{lemma_3_1} requires the number of undecided nodes to be inside the interval $\left[\frac{n}{3}\frac{1-4\epsilon}{1+6\epsilon}, \frac{n}{2}\right]$. We will later take care of this issue with Lemmas \ref{lemma:symbreakbadstarting1} and \ref{lemma:symbreakbadstarting2}, showing that whenever this number is not within the above interval, in at most $\bigo(\log n)$ rounds it will lie in it. Furthermore, Lemma \ref{lemma_3_1} guarantees that the condition on the undecided nodes holds ``only'' w.h.p., while   Lemma \ref{lemma:symmetrygeneric} requires this condition to  hold with probability $1$. We   show  this issue can be solved using a coupling argument similar to that in \cite{DBLP:conf/mfcs/ClementiGGNPS18}. The key point is that, starting from any configuration $\mathbf{x}$ with $q(\mathbf{x}) \in \left[\frac{n}{3}\frac{1-4\epsilon}{1+6\epsilon}, \frac{n}{2}\right]$, the probability that the process goes in one of those ``bad'' configurations with $q$ outside the above interval is negligible. Intuitively speaking, the configurations \emph{actually visited} by the process before breaking symmetry do satisfy the hypothesis of Lemma \ref{lemma:symmetrygeneric}. In order to make this argument rigorous, we define a  \emph{pruned} process, by removing all the \emph{unwanted} transitions.

Let $\Bar{s}\in \{0,1,\dots, n\}$, and $\mathbf{z}(\Bar{s})$ the configuration such that $s(\mathbf{z}(\Bar{s})) = \Bar{s}$, and $q(\mathbf{z}(\Bar{s})) = \frac{n}{2}$. Let $p_{\mathbf{x},\mathbf{y}}$ be the probability of a transition from the configuration $\mathbf{x}$ to the configuration $\mathbf{y}$ in the \uproc. The \pruned behaves exactly as the original process but every transition from a configuration $\mathbf{x}$ such that $q(\mathbf{x}) \in \left[\frac{n}{3}\frac{1-4\epsilon}{1+6\epsilon}, \frac{n}{2}\right]$ and $s(\mathbf{x}) = \bigo(\sqrt{n\log n})$ to a configuration $\mathbf{y}$ such that $q(\mathbf{y}) < \frac{n}{3}\frac{1-4\epsilon}{1+6\epsilon}$ or $q(\mathbf{y}) > \frac{n}{2}$ has probability $p_{\mathbf{x},\mathbf{y}}'=0$. Moreover, for any $\Bar{s}\in[n]$, starting from the configuration $\mathbf{x}$, the probability of reaching the configuration $\mathbf{z}(\Bar{s})$ is
\[p_{\mathbf{x},\mathbf{z}(\Bar{s})}'=p_{\mathbf{x},\mathbf{z}(\Bar{s})}+\sum_{\substack{\mathbf{y}:\ s(\mathbf{y}) = \Bar{s} \text{ and} \\ q(\mathbf{y}) \notin \left[\frac{n}{3}\frac{1-4\epsilon}{1+6\epsilon}, \frac{n}{2}\right] }}p_{\mathbf{x},\mathbf{y}}.\]
All the other transition probabilities remain the same. Observe that the \uproc is defined in such a way that it has exactly the same marginal probability of the original process w.r.t. the random variable $s(\mathbf{X}_t)$; thus, Lemma \ref{lemma:symbreak} holds for the \pruned as well and we can apply Lemma \ref{lemma:symmetrygeneric}. Then, the \pruned reaches a configuration having bias $\Omega(\sqrt{n\log n})$ within $\bigo(\log n)$ rounds, w.h.p., as shown in the following lemma.

\begin{lemma}\label{lemma:prunedwins}
       Starting from any configuration $\mathbf{x}$ such that $q(\mathbf{x}) \in \left[\frac{n}{3}\frac{1-4\epsilon}{1+6\epsilon}, \frac{n}{2}\right]$ and $s(\mathbf{x}) = \bigo(\sqrt{n\log n})$, the \pruned reaches a configuration having bias $\Omega(\sqrt{n\log n})$ within $\bigo(\log n)$ rounds, w.h.p.
\end{lemma}
\begin{proof}[Proof of Lemma \ref{lemma:prunedwins}]
    Let $\gamma>0$ be a constant and $m=\gamma\sqrt{n}\log n$ be the target value of the bias in Lemma \ref{lemma:symmetrygeneric}. Since $q(\mathbf{x}) \in \left[\frac{n}{3}\frac{1-4\epsilon}{1+6\epsilon}, \frac{n}{2}\right]$ and $s(\mathbf{x}) = \bigo(\sqrt{n\log n})$, the \pruned satisfies Lemma \ref{lemma:symbreak} with probability 1, and thus we can apply Lemma \ref{lemma:symmetrygeneric} (setting the function $f(\mathbf{X}_t) = s(\mathbf{X}_t)$), which gives us that the \pruned process reaches a configuration $\mathbf{y}$ having bias $s(\mathbf{y})\ge m = \Omega(\sqrt{n\log n})$ within $\bigo(\log n)$ rounds, w.h.p.
\end{proof}

We now want to go back to the original process. The definition of the \pruned suggests a natural coupling between it and the original one. If the two process are in different states, then they act independently, while, if they are in the same state $\mathbf{x}$, they move together unless the \uproc goes in a configuration $\mathbf{y}$ such that $q(\mathbf{y}) \notin \left[\frac{n}{3}\frac{1-4\epsilon}{1+6\epsilon}, \frac{n}{2}\right]$. In that case, the \pruned goes in $\mathbf{z}(s(\mathbf{y}))$. In the proof of the next lemma, we show that the time the \pruned takes to reach bias $\Omega(\sqrt{n\log n})$ stochastically dominates the one of the original process, giving the result.

\begin{lemma}\label{lemma:prunedtooriginal}
Starting from any configuration $\mathbf{x}$ such that $q(\mathbf{x}) \in \left[\frac{n}{3}\frac{1-4\epsilon}{1+6\epsilon}, \frac{n}{2}\right]$ and $s(\mathbf{x}) = \bigo(\sqrt{n\log n})$, the \uproc reaches a configuration having bias $\Omega(\sqrt{n\log n})$ within $\bigo(\log n)$ rounds, w.h.p.
\end{lemma}
\begin{proof}[Proof of Lemma \ref{lemma:prunedtooriginal}]
    Let $\{\mathbf{X}_t\}$ and $\{\mathbf{Y}_t\}$ be the original process and the pruned one, respectively. Call $H$ the set of possible initial configuration according to the hypothesis, and let $\mathbf{x}\in H$. Note that if $\mathbf{X}_t=\mathbf{Y}_t=\mathbf{x}$, then
    \[
    \mathbf{Y}_{t+1} = \begin{cases} 
      \mathbf{X}_{t+1} & \textbf{if } \mathbf{X}_{t+1} \in H \\
     \mathbf{z}(s(\mathbf{X}_t)) & \text{otherwise}
    \end{cases}.
    \]
    Let $\tau = \inf\{t: \nat : \abs{s(\mathbf{X}_t)} \ge \sqrt{n\log n}\}$, and let $\tau* = \inf\{t \in \nat : \abs{s(\mathbf{Y}_t)} \ge \sqrt{n\log n}\}$. For any configuration $\mathbf{x}\in H$, define $\rho_\mathbf{x}^t$ the event that the two processes $\{\mathbf{X}_t\}$ and $\{\mathbf{Y}_t\}$ have separated at round $t+1$, i.e.\ $\rho_\mathbf{x}^t = \{\mathbf{X}_t = \mathbf{Y}_t = \mathbf{x}_t\}\cap\{\mathbf{X}_{t+1} \neq \mathbf{Y}_{t+1}\}$. Observe that, if the two couple processes in the same configuration $\mathbf{x}_0\in H$ and $\tau > c\log n$, then either $\tau* > c\log n$ or there exists a round $t\le c\log n$ such that for some $\mathbf{x}\in H$ the event $\rho^t_{\mathbf{x}}$ has occurred. Hence, if $\prob'_{\mathbf{x}_0,\mathbf{x}_0}$ is the joint probability for the couple $(\mathbf{X}_t, \mathbf{Y_t})$ which both start at $\mathbf{x}_0$, we have
    \begin{align*}
        & \ \prob'_{\mathbf{x}_0,\mathbf{x}_0}(\tau > c\log n) \\
        \le & \ \prob'_{\mathbf{x}_0,\mathbf{x}_0}\left(\{\tau* > c\log n\}\cup \{\exists t\le c\log n, \exists \mathbf{x}\in H : \rho_\mathbf{x}^t\}\right) \\
        \le & \ \prob'_{\mathbf{x}_0,\mathbf{x}_0}(\tau* > c\log n) + \prob'_{\mathbf{x}_0,\mathbf{x}_0}(\exists t\le c\log n, \exists \mathbf{x}\in H : \rho_\mathbf{x}^t).
    \end{align*}
    As for the first item, since Lemma \ref{lemma:symmetrygeneric} holds for the \pruned, we have that it is upper bounded by $1/n^{-\Theta(1)}$. As for the second term, we get that
    \begin{align*}
        \prob'_{\mathbf{x}_0,\mathbf{x}_0}(\exists t\le c\log n, \exists \mathbf{x}\in H : \rho_\mathbf{x}^t) \le & \ \sum_{t=1}^{c\log n}\prob'_{\mathbf{x}_0,\mathbf{x}_0}\left(\exists \mathbf{x}\in H : \rho^t_\mathbf{x}\right)\\
        = & \ \sum_{t=1}^{c\log n}\sum_{\mathbf{x}\in H}\prob'_{\mathbf{x}_0,\mathbf{x}_0}\left( \rho^t_\mathbf{x}\right)\\
        \le & \ \sum_{t=1}^{c\log n} \frac{n^2}{e^{-\Theta(n)}} \\
        \le & \ \frac{1}{n},
    \end{align*}
    where in the second inequality we used the probabilities computed in the proof of Lemma \ref{lemma_3_1}, and the fact that $\abs{H}$ is at most all the combinations of parameters $q$ and $s$.
\end{proof}

Now, we take care of those cases in which the starting configuration is such that $q \notin \left[\frac{n}{3}\frac{1-4\epsilon}{1+6\epsilon}, \frac{n}{2}\right]$. Indeed, if $q < \frac{n}{3}\frac{1-4\epsilon}{1+6\epsilon}$, the following holds.

\begin{lemma}\label{lemma:symbreakbadstarting1}
    Let $\mathbf{x}$ be any starting configuration such that $q(\mathbf{x}) \le \frac{n}{2}$, and $s(\config) \le \frac{2\epsilon}{(1+6\epsilon)^2}n$. Then, at the next round, it holds that  $q(\mathbf{x}) \in \left[\frac{n}{3}\frac{1-4\epsilon}{1+6\epsilon}, \frac{n}{2}\right]$, w.h.p.
\end{lemma}
\begin{proof}[Proof of Lemma \ref{lemma:symbreakbadstarting1}]
    Let $f(q)= \frac{3}{4}\left(\frac{1+6\epsilon}{n}\right)q^2 - \frac{1+6\epsilon}{2}q + \frac{5+6\epsilon}{12}n $. By Equation \eqref{expectation_Q_oknoise} we have that 
    \[
        f(q) - \frac{\epsilon}{(1+6\epsilon)}n \le \expect{Q\mid \config} \le f(q).
    \]
    Then, $f(q)$ has its maximum in one of the two boundaries, namely $q=0$ or $q = \frac{n}{2}$. Observe that \[
        f(0) = \frac{n}{2} - \frac{1-6\epsilon}{12}n < \frac{n}{2}
    \]
    since $\epsilon < \frac{1}{6}$. At the same time, we have that 
    \[
        f(n/2) = -\frac{1+6\epsilon}{16}n + \frac{5+6\epsilon}{12}n \le \frac{n}{2} - \frac{1-6\epsilon}{12}n < \frac{n}{2}.
    \]
    Thus, for the additive form of Chernoff bound (Theorem \ref{chernoff:additive} in Appendix \ref{tools}), we have that 
    \[
        \pr{Q \ge \frac{n}{2} \mid \config} \le \pr{Q \ge \expect{Q \mid \config} + \frac{1-6\epsilon}{12}n} \le \exp\left(-\frac{2(1-6\epsilon)}{144}n\right). 
    \]
    On the other hand, the function
    \[
        f(q) - \frac{\epsilon}{(1+6\epsilon)}n = \frac{3}{4}\left(\frac{1+6\epsilon}{n}\right)q^2 - \frac{1+6\epsilon}{2}q + \frac{5+6\epsilon}{12}n- \frac{\epsilon}{(1+6\epsilon)}n
    \]
    has its minimum in $\bar{q} = \frac{n}{3}$. Then
    \[
      f(\bar{q}) - \frac{\epsilon}{(1+6\epsilon)}n = - \frac{1+6\epsilon}{12}n + \frac{5+6\epsilon}{12}n- \frac{\epsilon}{(1+6\epsilon)}n = \frac{n}{3}\left(\frac{1+3\epsilon}{1+6\epsilon}\right),
    \]
    which is at most $\expect{Q \mid \config} - \frac{6\epsilon n}{3(1+6\epsilon)}$.
    From the additive form of Chernoff bound (Theorem \ref{chernoff:additive} in Appendix \ref{tools}), this implies the following.
    \[
        \pr{Q \le \frac{1-3\epsilon}{1+6\epsilon} \mid \config} \le \pr{Q \le \expect{Q \mid \config} - \frac{2\epsilon n}{1+6\epsilon} \mid \config} \le \exp\left(-\frac{8 \epsilon^2 n}{(1+6\epsilon)^2}\right),
    \]
    which, together with the previous result, gives the thesis.
\end{proof}

On the other hand, the number of undecided nodes decreases as long as it is more than $n/2$, w.h.p. The following lemma shows this behaviour.

\begin{lemma}\label{lemma:symbreakbadstarting2}
       Let $\mathbf{x}$ be any starting configuration such that $q(\mathbf{x}) > \frac{n}{2}$. Then, at the next round, it holds that $Q \le q\left(\frac{5}{6}+\epsilon\right)$, w.h.p.
\end{lemma}
\begin{proof}[Proof of Lemma \ref{lemma:symbreakbadstarting2}]
    Consider $f(q)= \frac{3}{4}\left(\frac{1+6\epsilon}{n}\right)q^2 - \frac{1+6\epsilon}{2}q + \frac{5+6\epsilon}{12}n$, which no less than $\expect{Q \mid \config}$. We see that, for $q > n/2$, the following is true:
    \[
        \tilde{f}(q) = f(q) - q\left(\frac{2}{3}+2\epsilon\right) \le 0.
    \]
    Indeed, 
    \[
        \tilde{f}(q) = \frac{3}{4}\left(\frac{1+6\epsilon}{n}\right)q^2 - \frac{7+30\epsilon}{6}q+\frac{5+6\epsilon}{12}n,
    \]
    which has its maximum in one of the two boundaries, namely $q=n/2$ and $q=n$. We compute the expression in these quantities.
    \begin{align*}
        & \tilde{f}(n/2) = \frac{(9+54\epsilon)n-(28+120\epsilon)n+(5+6\epsilon)n}{48} = -\frac{7+30\epsilon}{24}n < 0, \\
        & \tilde{f}(n) = \frac{(9+54\epsilon)n-(14+60\epsilon)n+(5+6\epsilon)n}{12} = 0.
    \end{align*}
    Then, we have that
    \[
        \expect{Q \mid \config} \le q\left(\frac{2}{3}+2\epsilon\right).
    \]
    The additive form of Chernoff bound (Theorem \ref{chernoff:additive} in Appendix \ref{tools}) implies that
    \begin{align*}
        \pr{Q \ge q\left(\frac{5}{6}+ \epsilon\right) \mid \config}  = & \ \pr{Q \ge q\left(\frac{2}{3}+ 2\epsilon\right) + q\left(\frac{1}{6} - \epsilon\right) \mid \config} \\
        \le & \ \pr{Q \ge \expect{Q \mid \config} + q\left(\frac{1}{6} - \epsilon\right) \mid \config} \\
        \le & \ \pr{Q \ge \expect{Q \mid \config} + \frac{n}{2}\left(\frac{1}{6} - \epsilon\right) \mid \config} \\
        \le & \ \exp\left(\frac{n}{2}\left(\frac{1}{6} - \epsilon\right)^2\right).
    \end{align*}
    which gives the thesis.
\end{proof}
Finally, we are ready to prove Theorem \ref{thm:symbre}.

\begin{proof}[Proof of Theorem \ref{thm:symbre}: Wrap-Up]
    Let $\gamma>0$ be a constant and $m=\gamma\sqrt{n}\log n$ be the  target value of the bias in Lemma \ref{lemma:symmetrygeneric}.
    Let $\mathbf{x}$ be any initial configuration having bias $\lvert s \rvert < m$. We have two cases.
    \begin{itemize}
        \item[(i)] If the number of undecided nodes is such that $q(\config) \in \left[\frac{n}{3}\frac{1-4\epsilon}{1+6\epsilon}, \frac{n}{2}\right]$, then Lemma \ref{lemma:prunedtooriginal} implies that the \uproc reaches a configuration having bias $\Omega(\sqrt{n\log n})$ in $\bigo(\log n)$ rounds, w.h.p.;
        \item[(ii)] else, if the starting configuration is such that $q(\config) \notin \left[\frac{n}{3}\frac{1-4\epsilon}{1+6\epsilon}, \frac{n}{2}\right]$, then, for Lemmas \ref{lemma:symbreakbadstarting1} and \ref{lemma:symbreakbadstarting2}, the \uproc reaches within $\bigo(\log n)$ rounds a configuration $\mathbf{y}$ having $q(\mathbf{y}) \in \left[\frac{n}{3}\frac{1-4\epsilon}{1+6\epsilon}, \frac{n}{2}\right]$, w.h.p. Then, either $s(\mathbf{y}) = \Omega(\sqrt{n\log n)}$, or we are in case (i). As Lemma \ref{lemma:whp_intersection} in the preliminaries implies, the intersection of $\bigo(\log n)$ events that hold w.h.p.\ is an event which holds w.h.p.
    \end{itemize}
    Then, Theorem \ref{theorem_almost_plurality} gives the desired result.
\end{proof}

\section{Simulations} \label{sec:exp}
We made computer simulations with  values of the input size $n$ ranging from $2^{10}$ to $2^{17}$,   and for noise probabilities of $p=1/12$, $p=1/8$, $p=1/7$, and $p= 1/5$. Besides confirming the phase transition predicted by our theoretical analysis, the outcomes     show this behaviour emerges even for reasonable  sizes (i.e. $n$) of the system. 
Indeed, we made the \udyn run for 400 rounds for the above values of $p$. 

\begin{table}[htb]
	\centering
	\begin{tabular}{|c|c|c|c|}
	    \hline
		\multirow{2}{*}{Size $n$} & \multicolumn{3}{c|}{Average times} \\ \cline{2-4}
		& $p=1/12$ & $ p = 1/8$ & $p=1/7$ \\ \hline
		$2^{10}$ & $24$ & Failed & Failed \\ \hline
		$2^{11}$ & $24$ & $39$ & Failed \\ \hline
		$2^{12}$ & $28$ & $41$ & Failed \\ \hline
		$2^{13}$ & $27$ & $53$ & Failed \\ \hline
		$2^{14}$ & $32$ & $52$ & $77$ \\ \hline
		$2^{15}$ & $32$ & $54$ & $88$ \\ \hline
		$2^{16}$ & $36$ & $57$ & $96$ \\ \hline
		$2^{17}$ & $39$ & $68$ & $103$ \\ \hline
	\end{tabular}
	\caption{The average time to reach a meta-stable almost-consensus phase.}
	\label{table:goodnoise}
\end{table}
\begin{table}[htb]
	\centering
	\begin{tabular}{|c|c|c|}
	    \hline
		\multirow{2}{*}{Size $n$} & \multicolumn{2}{c|}{$p=1/5$}\\ \cline{2-3}
		& Average time & Number of switches  \\ \hline
		$2^{10}$ & $1$ & $39$ \\ \hline
		$2^{11}$ & $4$ & $42$ \\ \hline
		$2^{12}$ & $7$ & $42$ \\ \hline
		$2^{13}$ & $10$ & $37$ \\ \hline
		$2^{14}$ & $14$ & $38$ \\ \hline
		$2^{15}$ & $18$ & $38$ \\ \hline
		$2^{16}$ & $22$ & $44$ \\ \hline
		$2^{17}$ & $27$ & $39$ \\ \hline
	\end{tabular}
	\caption{The average 
	time the bias goes below $10\sqrt{n\log n}$, and the number of switches.}
	\label{table:badnoise}
\end{table}
In the first three settings of $p$, we started from complete balanced configurations (i.e. when both opinions  are supported by, respectively, $\frac n2$ agents) we found a fast convergence to the meta-stable regime of almost consensus, which then  did not break for all the rest of the simulation. Furthermore, we have noticed that the symmetry is always broken when the bias is ``roughly'' $10\sqrt{n\log n}$. 
As for the case  $p=1/5$, we started from a configuration of complete consensus  and we observed that, within a  short time, the system looses any information on the majority opinion (say, the bias becomes less than $10\sqrt{n\log n}$) and it keeps  this meta-stable phase with  many switches of the  majority opinion.
In Table \ref{table:goodnoise}, we can see the average time (computed over 100 trials and approximated to the  closest integer) in which the system enters the predicted meta-stable phase of almost consensus for any value of $p=1/12$, $p=1/8$, and $p=1/7$, for different input sizes. We also see that, when  $p$ gets close  to $1/6$,  the  emergent behaviour is observed only for large
values of    $n$ and  some of the experiments fail.
In Table \ref{table:badnoise}, we see the average times  in which the bias of the system goes below $10\sqrt{n\log n}$ for different input sizes, and the corresponding number of switches of majority opinion during the remaining time.

\section{Conclusions} \label{sec:concl}
While our mathematical analysis for the \undecided dynamics does not directly apply to other opinion dynamics, it suggests that a general phase-transition phenomenon may hold for a large class of dynamics characterized by an \emph{exponential drift} towards consensus configurations.
Our work thus  naturally poses the general question of whether it is possible to provide a characterization of
opinion dynamics with stochastic interactions,
in terms of their critical behavior with respect to uniform communication noise. 
 
As for the specific mathematical questions that follow from our results,  our assumption of a complete topology as underlying  graph is,  for several real MAS, a rather strong condition. 
However, two remarks on this issue follow. 
On one hand, we observe that, according to the adopted communication model,  at every round, every agent can pull information from just one other agent:
the \emph{dynamic} communication pattern is thus random and sparse. 
This setting may model opportunistic MAS where mobile agents use to meet randomly, at a relatively-high rate. 
On the other hand, we believe that a similar transition phase
does hold even for sparse topologies
having good expansion$/$conductance \cite{hoory_expander_2006}: this is an interesting   question left open by this work.

\appendix
\section{Appendix: Useful Tools}\label{tools}
Here we present the concentration results we have used all over the analysis. For an overview on the forms of Chernoff bounds see \cite{chernoff1} or \cite{chernoff2}.
\begin{theorem}[Multiplicative forms of Chernoff bounds]\label{chernoff:multiplicative}
	Let $X_1, X_2, \dots, X_n$ be independent $\{0,1\}$ random variables. Let $X = \sum_{i=1}^n X_i$ and $\mu=\mathbb{E}[X]$. Then:
	\begin{enumerate}
		\item[(i)] for any $\delta \in (0,1)$ and $\mu \le \mu_+ \le n$, it holds that 
		\begin{equation}\label{MCB+}
		    P\big(X\ge (1+\delta)\mu_+\big)\le e^{-\frac{1}{3}\delta^2\mu_+},
		\end{equation}
		\item[(ii)] for any $\delta \in (0,1)$ and $0 \le \mu_- \le \mu$, it holds that 
		\begin{equation}\label{MCB-}
			P\big(X\le (1-\delta)\mu_-\big)\le e^{-\frac{1}{2}\delta^2\mu_-}.
		\end{equation}
	\end{enumerate}
\end{theorem}
	
\begin{theorem}[Additive forms of Chernoff bounds]\label{chernoff:additive}
	Let $X_1, X_2, \dots, X_n$ be independent $\{0,1\}$ random variables. Let $X = \sum_{i=1}^n X_i$ and $\mu=\mathbb{E}[X]$. Then:
	\begin{enumerate}
		\item[(i)] for any $0 < \lambda < n$ and $\mu \le \mu_+ \le n$, it holds that 
		\begin{equation}\label{ACB+}
			P\big(X\ge \mu_+ +\lambda \big)\le e^{-\frac{2}{n}\lambda^2},
		\end{equation}
		\item[(ii)] for any $0 < \lambda < \mu_-$ and $0 \le \mu_- \le \mu$, it holds that 
		\begin{equation}\label{ACB-}
			P\big(X\le \mu_- - \lambda \big)\le e^{-\frac{2}{n}\lambda^2}.
		\end{equation}
	\end{enumerate}
\end{theorem}

The Berry-Eseen theorem is well treated in \cite{korolev2010}, and it gives an estimation on ``how far'' is the distribution of the normalized sum of i.i.d.\ random variables to the standard normal distribution.

\begin{theorem}[Berry-Eseen]\label{thm:berryeseen}
    Let $X_1, \dots, X_n$ be $n$ i.i.d.\ (either discrete or continuous) random variables with zero mean, variance $\sigma^2>0$, and finite third moment. Let $Z$ the standard normal random variable, with zero mean and variance equal to 1. Let $F_n(x)$ be the cumulative function of $ \frac{S_n}{\sigma\sqrt{n}}$, where $S_n = \sum_{i=1}^n X_i$, and $\Phi(x)$ that of $Z$. Then, there exists a positive constant $C>0$ such that
    \[
        \sup_{x\in \realnum} \abs{F_n(x) - \Phi(x)} \le \frac{C}{\sqrt{n}}
    \]
    for all $n\ge 1$.
\end{theorem}

\section{Appendix: Proofs}\label{app:preliminaries}


\begin{proof}[Equations \ref{expectation_S} and \ref{expectation_Q}]
    \begin{align*}
        \mean\left[S\bigm| \mathbf{x} \right] = \ & \mean\left[A\bigm| \mathbf{x} \right] - \mean\left[B\bigm| \mathbf{x} \right]\nonumber
    \\ = \ & \frac{1-2p}{n}\left[s(a+b)+2qs\right] + \frac{p}{n}\left[s(a+b)\right]\nonumber
    \\ = \ & s\frac{1-2p}{n}\left[n+q\right] + s\frac{p}{n}\left[n-q\right] \nonumber
    \\ = \ & s\left(1-p+(1-3p)\frac{q}{n}\right),\\
    \mean\left[Q\bigm| \mathbf{x} \right] = \ & + \frac{a}{n}\left[p(a+q)+(1-2p)b\right] \nonumber
    \\ & + \frac{b}{n}\left[p(b+q)+(1-2p)a\right] \nonumber
    \\ & + \frac{q}{n}\left[p(a+b)+(1-2p)q\right] \nonumber
    \\ = \ & \frac{p}{n}\left[a^2+b^2+2q(a+b)\right] \nonumber
    \\ & + \frac{1-2p}{n}\left[2ab+q^2\right] \nonumber
    \\ = \ & pn + \frac{1-3p}{n}\left[2ab+q^2\right] \nonumber
    \\ = \ & pn+\frac{1-3p}{2n}\left[2q^2+(n-q)^2-s^2\right].
    \end{align*}
\end{proof}

\begin{proof}[Proof of Lemma \ref{lem:equiv}]
    The equivalence between the two processes is showed through a coupling. Formally, consider the complete graph of $n$ nodes, namely $K_n$, over which the former process runs. We define another graph $G_n$, which contains a sub-graph isomorphic to $K_n$ in the following way. Let $K_n'$ be a copy of $K_n$, and let $H$ be a graph of $\additionalNodes =\frac{\pnoise}{1-\pnoise}n$ isolated nodes (which will be the \emph{stubborn agents)}. Then, each node $u\in H$ is connected by edges to all nodes of $K_n'$, namely, each node of $K_n'$ has as its neighborhood the whole set of nodes of $K_n' \cup H$, while each node of $H$ has as its neighborhood only the set of nodes of $K_n'$. The nodes of $H$ are such that $\additionalNodes \cdot p_{1}$ are stubborn agents supporting opinion 1,  $\additionalNodes \cdot p_{2}$ are stubborn agents supporting opinion 2, and so on. Observe that $\sum_{i=1}^m p_{i} = 1$, so this partition is well defined. 
    
    The \udyn behaves in exactly the same way over $G_n$, with the exception that the stubborn agents never change their opinion and that there is no noise perturbing communications between agents. The coupling is any bijective function $f: K_n \to K_n'$ such that, for each $v\in K_n$,  $v$ and $f(v)$ support the same opinion at the beginning of the process. Consider the two resulting Markov processes $\{\mathbf{X}_t\}_{t\ge 0}$ over $K_n$ and $\{\mathbf{X}_t'\}_{t\ge 0}$ over $K_n'$, denoting the opinion configuration at time $t$ in $K_n$ and in $K_n'$, respectively. It is easy to see that the two transition matrix are exactly the same (this is the meaning of \emph{equivalence} between the two processes). Indeed, in the former model (a), the probability an agent pulls opinion $j$ at any given round is
    \[
        (1-\pnoise)\frac{c_j}{n} + \pnoise \cdot p_j\ ,
    \]
    where $c_j$ is the size of the community of agents supporting opinion $j$; in the model defined in (b), the probability a non-stubborn agent pulls opinion $j$ at any given round is 
    \[
        \frac{c_j+\additionalNodes\cdot p_j}{n+\additionalNodes} = \frac{c_j+\frac{\pnoise}{1-\pnoise}n\cdot p_j}{n+\frac{\pnoise}{1-\pnoise}n} = (1-\pnoise)\cdot\frac{c_j}{n}+\pnoise\cdot p_{j}\ .
    \]
    
\end{proof}

\bibliographystyle{abbrv}

\bibliography{biblio}

\end{document}